\def\ARXIV{1}

\newif\ifcameraready
\ifdefined\ARXIV
  \camerareadyfalse
\else
  \camerareadytrue
\fi

\ifcameraready
  \documentclass[sigconf,acmthm=false]{acmart}
\else
  \documentclass[sigconf,nonacm,acmthm=false]{acmart}
\fi

\ifcameraready
  \copyrightyear{2026}
  \acmYear{2026}
  \setcopyright{cc}
  \setcctype{by-nc-nd}
  \acmConference[CCS '26]{Proceedings of the 2026 ACM SIGSAC Conference
    on Computer and Communications Security}{November 15--19, 2026}
    {The Hague, Netherlands}
  \acmBooktitle{Proceedings of the 2026 ACM SIGSAC Conference on Computer
    and Communications Security (CCS '26), November 15--19, 2026,
    The Hague, Netherlands}
  \acmDOI{10.1145/3830454.3832749}
  \acmISBN{979-8-4007-2871-6/2026/11}
\else
  \renewcommand\footnotetextcopyrightpermission[1]{}
\fi

\usepackage{graphicx}
\graphicspath{ {images/} }
\usepackage[font=small,skip=0pt]{caption}
\usepackage[font=small,skip=0pt]{subcaption}

\usepackage{algorithm}
\usepackage[noend]{algpseudocode}
\usepackage{mdwmath,mathtools}
\usepackage{mdwtab}
\usepackage[utf8]{inputenc}
\usepackage[english]{babel}
\usepackage{breqn}
\usepackage{tabto}
\usepackage{stmaryrd}
\usepackage{hyperref}

\usepackage{amsmath,amsfonts}
\usepackage{paralist}
\usepackage[normalem]{ulem}
\usepackage[T1]{fontenc}
\usepackage{booktabs}
\usepackage{multirow}
\usepackage{threeparttable}
\usepackage{multicol}
\usepackage[symbol]{footmisc}
\usepackage{makecell}

\usepackage{color}
\usepackage{url}
\usepackage{bbm}
\usepackage{bm}
\usepackage{graphics}
\usepackage[inline]{enumitem}
\usepackage{tikz-cd}
\usepackage{pgfplots}
\usepgfplotslibrary{groupplots}

\usepackage{tabularx}
\usepackage{xparse}
\usepackage[capitalise]{cleveref}
\usepackage{aliascnt}

\NewDocumentCommand{\statcirc}{ O{#2} m }{%
    \begin{tikzpicture}
    \fill[#2] (0,0) circle (1.0ex); 
    \fill[#1] (0,0) -- (180:1ex) arc (180:0:1ex) -- cycle; 
    \end{tikzpicture}
}

\usepackage{pifont}

\algdef{SE}[EVENT]{Upon}{EndUpon}[1]{\textbf{upon}\ #1\ \algorithmicdo}{\algorithmicend\ \textbf{event}}%
\algtext*{EndUpon}

\usepackage[framemethod=tikz]{mdframed} 
\usepackage{tablefootnote} 
\makeatletter 
\AfterEndEnvironment{mdframed}{%
 \tfn@tablefootnoteprintout%
 \gdef\tfn@fnt{0}%
}
\makeatother 

\usepackage{pgfplots}
\usetikzlibrary{patterns}
\pgfplotsset{compat=1.18}

\usepackage{xspace}

\usepackage{tikz}
\usetikzlibrary{arrows.meta,calc,positioning}

\tikzset{
    timeline/.style={gray!55, line width=0.4pt},
    timegrid/.style={gray!55, densely dotted, line width=0.35pt},
    msgone/.style={-Stealth, line width=0.9pt, blue},
    msgtwo/.style={-Stealth, line width=0.9pt, dashed, green!55!black},
    msgthr/.style={-Stealth, line width=0.9pt, dotted, red!55!black},
    msglabel/.style={midway, sloped, above, font=\scriptsize, inner sep=1pt},
    decision/.style={circle, draw=black, fill=white, line width=1.1pt,
                    minimum size=7.5pt, inner sep=0pt},
}

\newcommand{\lbfigsize}{0.31\linewidth}

\newcommand{\ExecDiagram}[1]{%
    \begin{tikzpicture}[x=0.2\linewidth,y=0.86cm]
        \def\yA{3}
        \def\yB{2}
        \def\yC{1}
        \def\yD{0}
        \def\xmax{3.25}
        \def\xname{-1.2em}
        
        \foreach \name/\val/\y in {A/$\vone$/\yA,B/$\vzero$/\yB,C/$\vone$/\yC,D/$\vzero$/\yD}{
            \node[anchor=center,font=\small] at (\xname,\y) {\name(\val)};
            \draw[timeline] (0,\y) -- (\xmax,\y);
        }
        
        \draw[-Stealth, line width=0.9pt] (0,-0.35) -- (\xmax,-0.35)
        node[anchor=west,font=\small] {time};
        
        \foreach \x/\lab in {1/{$\dhat$},2/{$2\dhat$}}{
            \draw[timegrid] (\x,-0.35) -- (\x,3.35);
            \node[anchor=north,font=\scriptsize] at (\x,-1em) {\lab};
        }
    
        \newcommand{\MarkByz}[1]{%
            \draw[red!90!black, line width=0.1pt] (0,##1) -- (\xmax,##1);
            \node[
                draw=red!90!black,
                circle,
                line width=0.7pt,
                inner sep=0.8em,
            ] at (\xname,##1) {};
        }
        #1
    \end{tikzpicture}%
}

\AtEndPreamble{%

  \theoremstyle{acmplain}
  \newtheorem{theorem}{Theorem}[section]
  \crefname{theorem}{theorem}{theorems}
  \Crefname{theorem}{Theorem}{Theorems}

  \newaliascnt{conjecture}{theorem}
  
  \aliascntresetthe{conjecture}
  \crefname{conjecture}{conjecture}{conjectures}
  \Crefname{conjecture}{Conjecture}{Conjectures}

  \newaliascnt{proposition}{theorem}
  
  \aliascntresetthe{proposition}
  \crefname{proposition}{proposition}{propositions}
  \Crefname{proposition}{Proposition}{Propositions}

  \newaliascnt{lemma}{theorem}
  \newtheorem{lemma}[lemma]{Lemma}
  \aliascntresetthe{lemma}
  \crefname{lemma}{lemma}{lemmas}
  \Crefname{lemma}{Lemma}{Lemmas}

  \newaliascnt{corollary}{theorem}
  \newtheorem{corollary}[corollary]{Corollary}
  \aliascntresetthe{corollary}
  \crefname{corollary}{corollary}{corollaries}
  \Crefname{corollary}{Corollary}{Corollaries}

  \newaliascnt{observation}{theorem}
  \newtheorem{observation}[observation]{Observation}
  \aliascntresetthe{observation}
  \crefname{observation}{observation}{observations}
  \Crefname{observation}{Observation}{Observations}

  \theoremstyle{acmdefinition}

  \newaliascnt{definition}{theorem}
  \newtheorem{definition}[definition]{Definition}
  \aliascntresetthe{definition}
  \crefname{definition}{definition}{definitions}
  \Crefname{definition}{Definition}{Definitions}

  \newaliascnt{example}{theorem}
  
  \aliascntresetthe{example}
  \crefname{example}{example}{examples}
  \Crefname{example}{Example}{Examples}
}

\newcommand{\cH}{\mathcal{H}}

\newcommand{\cF}{\mathcal{F}}

\newcommand{\mcp}[1]{{\sf mcp}(#1)}
\newcommand{\mce}[1]{{\sf mce}(#1)}
\newcommand{\tpref}[2]{{\sf tpref}_{#1}(#2)}
\newcommand{\mcptag}{\sf p}
\newcommand{\mcetag}{\sf e}

\newcommand{\pc}{Prefix Consensus\xspace}
\newcommand{\pccon}{Consistent \pc}
\newcommand{\pcver}{Verifiable \pc}
\newcommand{\pcacc}{Accountable \pc}
\newcommand{\pcstr}{Strong \pc}

\newcommand{\fullversionref}[1]{%
  \ifcameraready
    the full version~\cite{xiang2026prefix}%
  \else
    \Cref{#1}%
  \fi
}
\newcommand{\pcverstr}{Verifiable \pcstr}
\newcommand{\pcaccstr}{Accountable \pcstr}

\newcommand{\smr}{Multi-slot Consensus\xspace}

\newcommand{\bc}{Binary Consensus\xspace}

\newcommand{\qc}{\ensuremath{QC}\xspace}
\newcommand{\qcone}{\ensuremath{QC_1}\xspace}
\newcommand{\qctwo}{\ensuremath{QC_2}\xspace}
\newcommand{\qcthree}{\ensuremath{QC_3}\xspace}

\newcommand{\voteone}{\text{\sf vote-1}\xspace}
\newcommand{\votetwo}{\text{\sf vote-2}\xspace}
\newcommand{\votethree}{\text{\sf vote-3}\xspace}

\newcommand{\x}{\ensuremath{x}\xspace}
\newcommand{\xmcp}{\ensuremath{\x_{\mcptag}}\xspace}
\newcommand{\xmce}{\ensuremath{\x_{\mcetag}}\xspace}
\newcommand{\xmcpmcp}{\ensuremath{\x_{\mcptag\mcptag}}\xspace}
\newcommand{\xmcpmce}{\ensuremath{\x_{\mcptag\mcetag}}\xspace}

\newcommand{\y}{\ensuremath{y}\xspace}

\newcommand{\prefix}{\ensuremath{P}\xspace}
\newcommand{\sig}{\ensuremath{sig}\xspace}
\newcommand{\ver}{\ensuremath{ver}\xspace}

\newcommand{\hash}{\ensuremath{H}\xspace}

\newcommand{\sigsize}{\ensuremath{\kappa_{s}}\xspace}
\newcommand{\hashsize}{\ensuremath{\kappa_{h}}\xspace}
\newcommand{\const}{\ensuremath{c}\xspace}

\newcommand{\newview}{\text{\sf new-view}\xspace}
\newcommand{\proposal}{\text{\sf proposal}\xspace}
\newcommand{\timer}{\text{\sf timer}\xspace}

\newcommand{\cert}{\ensuremath{\mathsf{cert}}\xspace}

\newcommand{\newcommit}{\text{\sf new-commit}\xspace}

\newcommand{\emptyview}{\text{\sf empty-view}\xspace}

\newcommand{\rank}{\ensuremath{\mathsf{rank}}\xspace}

\newcommand{\id}{\ensuremath{\mathsf{id}}\xspace}
\newcommand{\sid}{\ensuremath{\mathsf{sid}}\xspace}

\newcommand{\abs}[1]{\lvert #1 \rvert}

\begin{document}

\title{Prefix Consensus For Censorship Resistant BFT}


\author{Zhuolun Xiang}
\affiliation{%
  \institution{Aptos Labs}
  \city{Palo Alto}
  \country{United States}
}
\email{daniel@aptoslabs.com}

\author{Andrei Tonkikh}
\affiliation{%
  \institution{Aptos Labs}
  \city{Palo Alto}
  \country{United States}
}
\email{andrei.tonkikh@aptoslabs.com}

\author{Alexander Spiegelman}
\affiliation{%
  \institution{Aptos Labs}
  \city{Palo Alto}
  \country{United States}
}
\email{sasha@aptoslabs.com}

\renewcommand{\shortauthors}{Xiang et al.}


\begin{abstract}
Despite broad adoption of BFT consensus in blockchains, censorship resistance remains weak: existing designs offer limited inclusion guarantees and allow leaders to exclude transactions.

We address this with a new abstraction and protocol stack. We define \emph{Prefix Consensus}, where parties input vectors and output two consistent vectors $(v^{\sf low},v^{\sf high})$ that extend the maximum common prefix of honest inputs and satisfy $v_i^{\sf low}\preceq v_j^{\sf high}$ for all honest parties $i,j$. We show that Prefix Consensus is solvable asynchronously and establish tight round-complexity bounds.

We then define \emph{Strong Prefix Consensus}, which additionally requires agreement on the high output, and give a leaderless partially synchronous protocol. Using its accountable variant, we build a leaderless, multi-proposer, censorship-resistant BFT SMR protocol with amortized four-round commit latency under synchronized starts, while guaranteeing that after GST at most $f$ slots can be censored.

Finally, we connect Prefix Consensus to graded consensus, obtaining a matching lower bound and a 3-round protocol, and derive leaderless Binary Consensus with improved worst-case complexity.
\end{abstract}

\ifcameraready
\fi


\maketitle

\ifcameraready\else
  \pagestyle{plain}
\fi

\section{Introduction}
\label{sec:intro}

Over the past decade, the rapid growth of blockchain systems has brought Byzantine Fault Tolerant (BFT) consensus to the forefront of both theoretical and practical research. On the theoretical side, the classical problem of
$n$ parties, each with an input, agreeing on a single common output has been extensively studied across a wide range of communication and adversarial models, yielding tight upper and lower bounds on round complexity, message complexity, and fault tolerance~\cite{lamport1982byzantine,dolev1985bounds,castro1999practical,kuznetsov2021revisiting}. On the practical side, modern BFT systems have demonstrated scalable throughput while achieving near optimal theoretical latency with optimal failure resilience~\cite{danezis2022narwhal,arun2024shoal++,giridharan2024autobahn}.     

Yet, despite the rapid growth of trading and other decentralized finance (DeFi) applications, censorship resistance has received comparatively little attention in the consensus literature. As a result, existing blockchain protocols provide no meaningful inclusion guarantees: a designated leader retains substantial discretion over which transactions are included in a block and which are excluded.

\paragraph{Leaderless vs multi-leader consensus}

Several recent works have explored leaderless~\cite{antoniadis2021leaderless,crain2018dbft} and multi-leader~\cite{stathakopoulou2019mir,keidar2021all,danezis2022narwhal,spiegelman2022bullshark,arun2024shoal++,babel2023mysticeti,giridharan2024autobahn,tonkikh2025raptr} consensus protocols in an effort to reduce or eliminate reliance on designated leaders. Informally, a leaderless protocol guarantees liveness even if an adversary can suspend an arbitrary single party at any round, implying that progress does not depend on the availability of a designated leader.
Multi-leader (sometimes referred to as multi-proposer) protocols, while not formally defined in the theoretical literature, are commonly used in practice to mitigate the performance bottlenecks of single-leader designs by allowing multiple parties to concurrently disseminate data.
Despite their names, however, to the best of our knowledge, neither leaderless nor multi-leader protocols inherently provide censorship resistance.

\subsection{Prefix Consensus}

We introduce \emph{Prefix Consensus}, a new consensus primitive designed 
to capture the minimal structure needed to reason about inclusion in adversarial settings,
in order to enable censorship-resistant BFT, which we believe is of independent interest. In Prefix Consensus, parties propose vectors of values and output low and high vectors such that every honest low extends the maximum common prefix of honest inputs and prefixes every honest high. Unlike classical consensus, parties need not agree on a single output.
More specifically, each honest party outputs two vectors, denoted $v^{\mathsf{low}}$ and $v^{\mathsf{high}}$, such that for any pair of honest parties $i,j$, the output $v^{\mathsf{high}}_j$ extends $v^{\mathsf{low}}_i$. Intuitively, $v^{\mathsf{low}}$ represents values that are safe to commit, while $v^{\mathsf{high}}$ represents values that are safe to extend in multi-slot consensus protocols.

Interestingly, and similarly to connected consensus~\cite{connected-cons-2024,connected-cons-2026}, Prefix Consensus is solvable in a purely asynchronous environment. This is possible because the primitive allows parties to output different values, as long as these values are mutually consistent and extend the common prefix of honest inputs. 

We formally define Prefix Consensus and study it in its own right, establishing tight upper and lower bounds on its round complexity under optimal resilience. In particular, we prove the following.

\begin{theorem}[Informal] \label{thm:main-informal}
Three communication rounds are necessary and sufficient to solve Prefix Consensus in an asynchronous Byzantine setting under optimal resilience.
\end{theorem}

Intriguingly, our results also imply (previously unknown) tight lower and upper bounds on the latency of Graded Consensus~\cite{connected-cons-2024,connected-cons-2026} and, to the best of our knowledge, improve the state-of-the-art algorithms by more than a factor of 2 in latency. We discuss it in detail in \fullversionref{sec:discussion:graded}.

\paragraph{\pcstr}
To use \pc as a consensus building block, we introduce \emph{Strong Prefix Consensus}, which strengthens \pc by requiring agreement on the high output. Our \pcstr construction is leaderless under partial synchrony and builds modularly on a verifiable variant of \pc: view~1 runs \pc on the instance input to prevent censorship, while later views run \pc on certificate vectors until some view commits a non-empty value that fixes a unique view-1 high output.
We also define accountable variants of Prefix Consensus and Strong Prefix Consensus, whose verifiable truncated outputs identify the responsible Byzantine proposer.

\subsection{Censorship Resistance}

Recent work~\cite{abraham2025latency} formally defines a strong notion of censorship resistance for \emph{single-shot} BFT consensus. At a high level, a single-shot protocol is censorship resistant if the agreed output contains the inputs of all honest parties. Abraham et al.~\cite{abraham2025latency} further prove that achieving this strong notion of censorship resistance requires at least five communication rounds (even assuming synchronized local clocks).

In this paper, we generalize and slightly relax this notion to the \emph{multi-slot} setting in order to obtain more efficient protocols. Informally, we say that a multi-slot consensus protocol achieves \emph{$c$-censorship resistance} if, after the network becomes synchronous (GST), at most $c$ slots are censored—i.e., at most $c$ slots fail to include the inputs of some honest parties.

Our multi-slot BFT protocol achieves \emph{$f$-censorship resistance}, where $f$ is the upper bound on the number of Byzantine parties. Intuitively, each Byzantine party can censor honest parties at most once before being detected and demoted. Since $f$ is fixed while the execution is long-lived, this bounded censorship is negligible in blockchain settings. 
This relaxation yields amortized four-round commit latency under synchronized starts.
In contrast to recent works~\cite{garimidi2025multiple,abraham2025latency} which still rely on designated leaders, our protocol is leaderless~\cite{antoniadis2021leaderless}.

\subsection{Technical Overview}

\paragraph{Prefix Consensus.}
Our Prefix Consensus protocol is a deterministic three-round asynchronous protocol under optimal resilience.
Parties exchange signed information to certify a large set of mutually consistent prefixes, and output a pair $(v^{\sf low},v^{\sf high})$ where $v^{\sf low}$ is safe to commit and $v^{\sf high}$ is safe to extend.
We also show this is tight: three rounds are necessary to solve Prefix Consensus in the asynchronous Byzantine setting for $n=3f+1$.

\paragraph{From Prefix Consensus to Strong Prefix Consensus.}
To obtain a consensus-like primitive, we strengthen Prefix Consensus by requiring agreement on the \emph{high} output.
Our construction is leaderless under partial synchrony and builds modularly on verifiable Prefix Consensus instances.
View~1 runs Prefix Consensus on the actual input and immediately commits the resulting low value.
Later views run Prefix Consensus on vectors of certificates derived from prior views; any committed \emph{non-empty} certificate vector induces a parent pointer chain back to view~1, which fixes a unique view-1 high value.

\paragraph{From Strong Prefix Consensus to Multi-slot Consensus.}
The Multi-slot Consensus protocol runs one Strong Prefix Consensus instance per slot, sequentially.
In each slot, parties broadcast their proposals, order received proposals into a length-$n$ hash vector using the current ranking (where ranking is some deterministic total ordering of all parties' indexes), and invoke Strong Prefix Consensus.
The low output is committed immediately; the agreed high output finalizes the slot and advances to the next slot.
Ranking updates use both the agreed high proof and truncated low proofs reported in the next slot, ensuring that all honest parties apply the same accountable demotions.



\paragraph{$f$-Censorship Resistance.}
In the \smr protocol, a Byzantine party may attempt to censor honest parties by sending different values to different honest parties or only sending its value to some of the honest parties. In this case, Strong Prefix Consensus (and underlying Prefix Consensus) guarantees that values appearing before the Byzantine party in the predefined ranking (called \emph{\rank}) are committed. Values after the equivocating party are excluded, making the malicious party explicit.
Censorship resistance is achieved by dynamically updating the ranking across slots. Agreed high evidence demotes a proposer responsible for censorship, while carried low evidence demotes a proposer responsible for delaying early commitment. Since the evidence is agreed before each update, all honest parties maintain the same ranking, and each Byzantine party can censor at most one slot.
In high-performance blockchains~\cite{aptos,sui,solana} producing multiple blocks per second, the fraction of affected blocks is negligible.

The key insight into why \pc allows us to solve Censorship Resistance is that we never need to ``ban'' any proposer.
We are able to provide a degree of isolation to the proposers simply by reordering them in the ranking.
In partially synchronous networks, it is impossible to determine with certainty if a node indeed misbehaved or if it simply failed to disseminate its message in time due to a temporary network asynchrony.
Thus, banning proposers would inevitably fail to achieve Censorship Resistance as any banned proposer must eventually be given the benefit of the doubt and unbanned or we would risk keeping an honest proposer banned forever.
Prefix Consensus resolves this tension.

\paragraph{Leaderless.} 
To satisfy leaderless termination, we allow special placeholder values (denoted $\bot$) in the vector inputs to (Strong) Prefix Consensus. At the beginning of each slot, honest parties wait for a bounded time to collect proposals; missing values are filled with $\bot$. From the perspective of (Strong) Prefix Consensus, $\bot$ is a valid value, and as long as all honest parties agree on the placement of $\bot$, the protocol commits the entire vector.
This prevents progress from depending on any single party’s timely participation, satisfying the leaderless property\cite{antoniadis2021leaderless}---No individual party can stall the slot by withholding its proposal.





\subsection{Roadmap}
We present a three-round asynchronous Prefix Consensus protocol for $n\ge 3f+1$ in \Cref{sec:pc}, and prove in \Cref{sec:lower-bound} that three rounds are necessary when $n\le 4f$.
We then introduce verifiable \pc in \Cref{sec:pcver}, define {\em \pcstr}, and give a leaderless partially synchronous construction in \Cref{sec:pcstr-leaderless}.
\Cref{sec:accountability} defines accountable variants of \pcver and
\pcstr.
Using accountable \pcstr as a black-box primitive, \Cref{sec:smr}
constructs a multi-slot BFT protocol achieving $f$-censorship resistance.
\ifcameraready
Additional analyses and results, including a communication-optimized
\pc implementation, a 2-round protocol for $n\ge 5f+1$, and the
connection to Graded Consensus, appear in the full
version~\cite{xiang2026prefix}.
\else
Additional results appear in the appendix, including the analyses of
\pcstr and \smr, a communication-optimized \pc implementation, a
2-round \pc protocol for $n\ge 5f+1$, and the connection to Graded
Consensus.
\fi

\section{Problem Definition}
\label{sec:def}

\subsection{Notation and Preliminary}
\label{sec:def:notation}

{\em Prefix.}
For a vector $\x$, let $\x[i]$ denote its $i$-th element and $\x[i:j]$ the subvector from $i$ to $j$.
We write $\y\preceq\x$ (equivalently $\x\succeq\y$) if $\abs{\y}\le \abs{\x}$ and $\y[i]=\x[i]$ for all $1\le i\le \abs{\y}$; in this case $\y$ is a \emph{prefix} of $\x$ and $\x$ an \emph{extension} of $\y$.
If additionally $\y\neq \x$, we write $\y\prec\x$ (or $\x\succ\y$).
Two vectors are \emph{consistent}, written $\x\sim\y$, if $\x\preceq \y$ or $\y\preceq \x$; otherwise they are \emph{conflicting}, written $\x\not\sim\y$.
W.l.o.g., all input-vector elements have the same size, denoted by $\const$.
For a set $S$ of vectors, define {\em maximum common prefix} $\mcp{S}=\max\{\x:\forall s\in S,\ \x\preceq s\}$ and {\em minimum common extension} $\mce{S}=\min\{\x:\forall s\in S,\ s\preceq \x\}$; if $\mce{S}$ does not exist, set $\mce{S}=\bot$.

For a multiset $S$ of equal-length vectors and a threshold
$t>|S|/2$, define the \emph{threshold prefix}
$\tpref{t}{S}=[z_1,\ldots,z_\ell]$, where $\ell$ is maximal such that,
for every $k\le\ell$,
\[
|\{v\in S:v[k]=z_k\}|\ge t.
\]
The value $z_k$ is unique because $2t>|S|$.
Equivalently, the threshold prefix stops at the first coordinate where no
value has $t$ supporters.

{\em Model.}
We assume the standard BFT model~\cite{lamport1982byzantine} with $n\ge 3f+1$ parties, of which at most $f$ may be Byzantine. A party is \emph{honest} if it is not Byzantine, and $\cH$ denotes the set of honest parties.
Channels are reliable and authenticated.
We consider both asynchronous and partially synchronous communication: in the asynchronous model, message delays are unbounded; in the partially synchronous model, there exists a Global Stabilization Time (GST) after which all messages are delivered within a known bound~$\Delta$.
Throughout the paper, GST denotes an \emph{effective} stabilization time:
it may be chosen later than the network's stabilization time to absorb any
finite protocol-specific stabilization and evidence-cleanup period.
We use $\delta$ to denote the actual message delay that is unknown to any party.
We refer to a protocol instance (or SMR slot) as having a
\emph{synchronized start} when all honest parties begin that instance (or
enter that slot) at the same time and GST has already occurred.

{\em Signatures.}
Let $\sig_i(x)$ denote party $i$'s signature on message $x$. Let $\ver_i(x,\sigma)$ denote the corresponding verification algorithm.
We assume domain separation: each message is signed together with a unique type tag (e.g., {\sf VOTE-1} for \voteone and {\sf VOTE-2} for \votetwo). For brevity, we omit tags in the presentation.

{\em Leaderless Termination.}
We use the notion of \emph{Leaderless Termination} from Antoniadis et al.~\cite{antoniadis2021leaderless}. 
The formal definition appears in the original paper; we give an informal description for intuition.

\begin{definition}[Leaderless Termination (informal)~\cite{antoniadis2021leaderless}]\label{def:leaderless}
A consensus protocol satisfies \emph{Leaderless Termination} if, after the network becomes synchronous, it is guaranteed to terminate even if an adversary can suspend up to $k\le f$ processes per round, potentially choosing different processes in different rounds, while up to $f-k$ other processes are Byzantine.
We focus on $k=1$: one suspended process per round and up to $f-1$ Byzantine processes.
\end{definition}

Clearly, a leader-based protocol does not satisfy Leaderless Termination since the adversary can suspend a party the round it acts as the leader.

\subsection{\pc}
\label{sec:def:pc}

We now define the core problem in the \pc family. 
For readability, we defer the definitions of the other variants 
to the sections where they are first used.

\begin{definition}[\pc]\label{def:pc}
    Each party $i$ has an input vector $v_{i}^{\sf in}$ of length $L$ and outputs a pair $(v_{i}^{\sf low}, v_{i}^{\sf high})$ of vectors of length at most $L$.
    A protocol solves \pc if it satisfies the following properties:
    \begin{itemize}
        \item \textbf{Prefix.} $v_{i}^{\sf low}\preceq v_{j}^{\sf high}$ for any honest parties $i,j\in\cH$.
        \item \textbf{Termination.} Every honest party eventually outputs.
        \item \textbf{Validity.} $\mcp{\{v_{h}^{\sf in}\}_{h\in \cH}}\preceq v_{i}^{\sf low}$ for any honest party $i\in\cH$.
    \end{itemize}
\end{definition}

We sometimes refer to $v_i^{\sf low}$ and $v_i^{\sf high}$ as party $i$'s \emph{low} and \emph{high} outputs (or values), or simply \emph{low} and \emph{high}.

\pc is strictly weaker than consensus: it can be solved deterministically under asynchrony without randomness (\Cref{sec:pc}), whereas this is impossible for consensus~\cite{fischer1985impossibility}.  
Prefix consensus is also closely related to Graded Consensus~\cite{connected-cons-2024}, as discussed in \fullversionref{sec:discussion:graded}.

The Prefix property immediately implies prefix consistency among the low outputs:
\begin{corollary}\label{cor:pc:1}
    $v_{i}^{\sf low}\sim v_{j}^{\sf low}$ for any $i,j\in \cH$.
\end{corollary}
\begin{proof}
By the Prefix property, for any honest $k\in\cH$ we have $v_{i}^{\sf low}\preceq v_{k}^{\sf high}$ and $v_{j}^{\sf low}\preceq v_{k}^{\sf high}$. Therefore $v_{i}^{\sf low}\sim v_{j}^{\sf low}$.
\end{proof}

It is also straightforward to show that $n=3f+1$ is the tight resilience threshold for \pc in the asynchronous Byzantine setting, via a standard indistinguishability/partitioning argument (as in the classic optimal-resilience proof for Byzantine agreement~\cite{lamport1982byzantine}). We omit the proof for brevity.

\begin{theorem}
    It is impossible to solve \pc in an asynchronous setting when $f\geq n/3$ parties can be Byzantine.
\end{theorem}

Now we define \pcstr, a natural extension of \pc that yields consensus on the high vector. 
Our \pc solution is a direct building block for \smr.

\begin{definition}[\pcstr]\label{def:pcstr}
    \pcstr satisfies all conditions of \pc (\Cref{def:pc}) and additionally ensures: 
    \begin{itemize}
        \item \textbf{Agreement}: $v_{i}^{\sf high} = v_{j}^{\sf high}$ for any honest parties $i,j\in\cH$.
    \end{itemize}
\end{definition}

\paragraph{Relations to other consensus primitives.}
\label{sec:def:pcstr:relation}
Observe that \pcstr can implement \bc as a special case: each party invokes \pcstr on a length-one vector, using input $[0]$ or $[1]$, and outputs the sole element of the agreed high value if it is non-empty, and a fixed default bit otherwise. Consequently, any lower bound for \bc immediately carries over to \pcstr.

Finally, since \pc is strictly weaker than (binary) consensus, it is also strictly weaker than \pcstr.

\subsection{\smr }
\label{sec:def:smr}

\begin{definition}[\smr]\label{def:smr}
\smr proceeds in discrete slots $s\in\mathbb{N}$. In each slot $s$, every party $i$ has an input value $b^{\sf in}_{i,s}$ and eventually commits an output vector $\mathbf{b}^{\sf out}_{i,s}$. A protocol solves \smr if, for every slot $s$, it satisfies:
\begin{itemize}
    \item \textbf{Agreement.} For any two honest parties $i,j\in\cH$, $\mathbf{b}^{\sf out}_{i,s}=\mathbf{b}^{\sf out}_{j,s}$.
    \item \textbf{Termination.} Every honest party $i\in\cH$ eventually commits $\mathbf{b}^{\sf out}_{i,s}$.
\end{itemize}
\end{definition}

In blockchain applications one typically also requires \emph{external validity}~\cite{cachin2001secure}, meaning that every element of each committed output vector satisfies a publicly known predicate. This can be enforced by requiring inputs to be externally valid and by locally discarding invalid proposals. For brevity, we omit it from the formal definition and focus on Agreement and Termination.

\subsection{Censorship Resistance}
\label{sec:def:censor}

\begin{definition}[Censorship Resistance]\label{def:censor}
Fix a slot $s$ that occurs after GST. A \smr protocol is \emph{censorship resistant in slot $s$} if every honest party’s input for slot $s$ appears in the committed output of every honest party, i.e.,
\[
    b^{\sf in}_{i,s} \in \mathbf{b}^{\sf out}_{j,s}
    \quad \text{for all } i,j \in \cH.
\]
If this condition fails, we say that slot $s$ is \emph{censored}.
More generally, a \smr protocol achieves \emph{$c$-Censorship Resistance} if, after GST, at most $c$ slots are censored.
\end{definition}

We treat each slot input as tagged by its slot and proposer, i.e.,
$b^{\sf in}_{i,s}\equiv(s,i,\mathsf{payload}_{i,s})$; thus identical
payloads from different proposers are distinct for censorship, although
these tags may be stripped before application execution.


{\em Comparison to~\cite{abraham2025latency,garimidi2025multiple}.}
Our definition is a bounded multi-slot relaxation of the single-shot notion
of~\cite{abraham2025latency}: lifting it slot-by-slot gives
$0$-Censorship Resistance, whereas we allow up to $c$ censored post-GST slots
while still ruling out infinitely many inclusion failures.
Unlike the all-or-nothing notion of~\cite{garimidi2025multiple}, which allows
empty slots, our definition counts any slot omitting an honest input as
censored. See~\Cref{sec:related} for a detailed comparison.

Any \smr protocol that satisfies $c$-Censorship Resistance for some finite $c$ can be used to implement standard BFT state machine replication~\cite{schneider1990implementing} with censorship resistance. Concretely, clients simply (re)submit transactions until they are included, and if necessary submit to multiple parties (or switch recipients) to ensure that at least one honest party receives the transaction. 
Since after GST at most $c$ slots can be censored, once $c$ censored slots have occurred, every transaction delivered to an honest party is included and ordered in the very next slot.

\paragraph{Latency of censorship resistant \smr.}
We define two latency measures for slots that start after GST.
As usual in partial synchrony, protocols may require an additional constant stabilization period after GST before making progress; we absorb this into GST.

\begin{definition}[Commit Latency]\label{def:commit-latency}
Fix a slot $s$ that starts after GST and is not censored.
The \emph{commit latency} of slot $s$ is the time from the start of slot $s$
to the earliest time $T$ such that, by time $T$, every honest party $j\in\cH$
has committed \emph{all honest inputs} for slot $s$, i.e.,
\[
\forall i,j\in\cH:\quad
b^{\sf in}_{i,s}\ \text{is committed by party $j$ by time $T$}.
\]
\end{definition}

\begin{definition}[Slot Latency]\label{def:slot-latency}
Fix a slot $s$ that starts after GST and is not censored.
The \emph{slot latency} of slot $s$ is the time from the start of slot $s$
to the earliest time $T$ such that, by time $T$, every honest party has started slot $s{+}1$.
\end{definition}

Commit latency can be strictly smaller than slot latency: a low output may commit all honest inputs before the agreed high finalizes the slot and triggers slot $s{+}1$.
In sequential-slot constructions this gap directly increases slot latency; \fullversionref{sec:smr:slot-latency} discusses pipelining techniques that reduce it.

\begin{definition}[Amortized Commit Latency]
Consider an infinite execution after GST, and let
$s_1,s_2,\ldots$ be the non-censored slots.  Let $L(s_k)$ denote the commit
latency of slot $s_k$.  The \emph{amortized commit latency} is
\[
    \limsup_{m\to\infty} \frac{1}{m}\sum_{k=1}^{m} L(s_k).
\]
A protocol has amortized commit latency at most $L$ if this quantity is at
most $L$ in every execution.
\end{definition}

\section{\pc}
\label{sec:pc}

In this section, we present a simple 3-round protocol for \pc (\Cref{def:pc}) with $O(n^2)$ message complexity and communication complexity $O((\const L + \sigsize)n^4)$ (\Cref{sec:pc:simple}). Although not communication-efficient, it serves as a warm-up and highlights the main ideas.
The first-round certification rule treats coordinates independently: a coordinate is appended only when one value has strict-majority support within the first-round quorum.
The round complexity of our \pc protocol is \emph{optimal}, as a matching lower bound is shown in~\Cref{sec:lower-bound}.

Due to space constraints, we defer the communication-optimized
implementation to \fullversionref{sec:pc:optimized}; it has complexity
$O(\const n^2 L+\const n^3+\sigsize n^2 L)$.

\subsection{3-Round \pc Protocol}
\label{sec:pc:simple}

For readability, when the context is clear we write $x\in\qc$ to mean that $\qc$ contains a tuple whose first component is $x$ (regardless of any additional fields). Under this convention, expressions such as $\mcp{\{x : (x,\ldots)\in\qc\}}$ are abbreviated as $\mcp{\{x\in\qc\}}$.

Whenever a party receives a message, it verifies the message before processing it.
In particular, it checks that every QC contains exactly $n-f$ valid
messages from distinct senders for the same protocol instance, that every
nested QC is valid, that the claimed value equals the output of the
corresponding certification procedure, and that every message is correctly
signed and domain-separated by its phase and instance.
Each party counts at most one message per sender, phase, and instance;
additional or equivocating messages from that sender are ignored.
Each honest party broadcasts at most one message in each phase.

\begin{algorithm}[t!]
\caption{\pc Protocol}
\label{alg:pc}
\begin{algorithmic}[1]
\Function{\sf QC1Certify}{\qcone}
    \State $t:=\lfloor(n-f)/2\rfloor+1$
    \State let $\mathcal V_1$ be the multiset of vectors in \qcone
    \State \Return $\x:=\tpref{t}{\mathcal V_1}$
\EndFunction

\Statex 

\Function{\sf QC2Certify}{\qctwo}
    \State \Return $\xmcp:=\mcp{\{\x \in \qctwo\}}$
    \Comment{maximum common prefix}
\EndFunction

\Statex 

\Function{\sf QC3Certify}{\qcthree}
    \State $\xmcpmcp:=\mcp{\{\xmcp \in \qcthree\}}$
    \Comment{maximum common prefix}
    \State $\xmcpmce:=\mce{\{\xmcp \in \qcthree\}}$
    \Comment{minimum common extension}
    \State \Return $(\xmcpmcp, \xmcpmce)$
\EndFunction

\Statex 

\Upon{input $v_{i}^{\sf in}$}
    \State broadcast $\voteone:=(v_{i}^{\sf in}, \sig_i(v_{i}^{\sf in}))$
\EndUpon

\Statex 

\Upon{first obtaining $n-f$ valid \voteone from distinct parties}
    \State let \qcone be the set of $n-f$ \voteone messages
    \State $\x:={\sf QC1Certify}(\qcone)$
    \State broadcast $\votetwo:=(\x, \sig_i(\x), \qcone)$
\EndUpon

\Statex 

\Upon{first obtaining $n-f$ valid \votetwo from distinct parties}
    \State let \qctwo be the set of $n-f$ \votetwo messages
    \State $\xmcp:={\sf QC2Certify}(\qctwo)$
    \State broadcast $\votethree:=(\xmcp, \sig_i(\xmcp), \qctwo)$
\EndUpon

\Statex 

\Upon{first obtaining $n-f$ valid \votethree from distinct parties}
    \State let \qcthree be the set of $n-f$ \votethree messages
    \State $(\xmcpmcp, \xmcpmce):={\sf QC3Certify}(\qcthree)$
    \State {\bf output} $v_{i}^{\sf low}:=\xmcpmcp$ and $v_{i}^{\sf high}:=\xmcpmce$
\EndUpon

\end{algorithmic}
\end{algorithm}

\subsubsection{Protocol Description.}\label{sec:pc:simple:protocol}

In the \pc protocol, each party proposes an input vector, and the protocol proceeds through three asynchronous rounds of voting. Each round the parties collect $n-f$ votes into a \emph{quorum certificate} (QC) and enter the next round. At a high level:

\begin{itemize}
    \item {\em Round 1 (voting on inputs).}
    Each party broadcasts its \voteone by signing the input vector.
    Let $q:=n-f$ and $t:=\lfloor q/2\rfloor+1$.
    From any $q$ received \voteone messages, a party computes their
    threshold prefix $\x$, which becomes its \votetwo.
    Different coordinates may be supported by different sets of voters,
    so $\x$ need not be a prefix of any single input vector.
    \item {\em Round 2 (voting on certified prefixes).}
    Parties sign and broadcast $\x$ as \votetwo.  From a quorum of \votetwo, a party derives a certified prefix $\xmcp$, which is the longest prefix shared by all \votetwo in the quorum.
    
    \item {\em Round 3 (deriving outputs).}
    Parties vote on $\xmcp$ as \votethree and obtain a third QC.  From it, they compute the outputs as the maximum common prefix and minimum common extension of all certified prefixes.
\end{itemize}

\subsection{Analysis of \pc Protocol}\label{sec:pc:simple:analysis}
In this section, we analyze the properties of our \pc protocol (\Cref{alg:pc}) presented in~\Cref{sec:pc}.

For brevity, we use the notation $\qcone.\x:={\sf QC1Certify}(\qcone)$, $\qctwo.\xmcp:={\sf QC2Certify}(\qctwo)$, and $(\qcthree.\xmcpmcp, \qcthree.\xmcpmce):={\sf QC3Certify}(\qcthree)$. 

\begin{lemma}\label{lem:pc:1}
    $\qctwo.\xmcp\sim\qctwo'.\xmcp$ for any $\qctwo, \qctwo'$. 
\end{lemma}
\ifcameraready
\begin{proof}
Deferred to \fullversionref{sec:pc:simple:analysis}.
\end{proof}
\else
\begin{proof}
    For any $\qctwo,\qctwo'$, by quorum intersection, at least one honest 
    party's $\votetwo=(\x, *, *)$ is in both $\qctwo$ and $\qctwo'$.
    Since $\qctwo.\xmcp= \mcp{\{\x\in \qctwo\}} \preceq\x$ and $\qctwo'.\xmcp = \mcp{\{\x \in \qctwo'\}} \preceq\x$, we conclude that $\qctwo.\xmcp\sim\qctwo'.\xmcp$. 
\end{proof}
\fi

\begin{lemma}\label{lem:pc:qc1-threshold}
Let $q=n-f$ and $t=\lfloor q/2\rfloor+1$.
Among any $q$ votes, at most one value can have support from at least $t$ distinct voters at a fixed coordinate.
Moreover, when $n\ge 3f+1$, every set of $q$ voters contains at least $t$ honest parties.
\end{lemma}
\ifcameraready
\begin{proof}
Deferred to \fullversionref{sec:pc:simple:analysis}.
\end{proof}
\else
\begin{proof}
Since $2t>q$, two disjoint sets of $t$ voters cannot both fit in a set of size $q$.
As each voter contributes only one value at a coordinate, at most one value can have $t$ supporters.

Any set of $q=n-f$ voters contains at least $n-2f$ honest parties.
The assumption $n\ge 3f+1$ implies
$n-2f\ge \lfloor(n-f)/2\rfloor+1=t$.
\end{proof}
\fi

\begin{lemma}\label{lem:pc:2}
    $\mcp{\{v_{h}^{\sf in}\}_{h\in \cH}}\preceq \qcone.\x$ for any \qcone.
\end{lemma}
\ifcameraready
\begin{proof}
Deferred to \fullversionref{sec:pc:simple:analysis}.
\end{proof}
\else
\begin{proof}
    Let $v_{\cH}=\mcp{\{v_{h}^{\sf in}\}_{h\in \cH}}$ and fix any
    coordinate $k\le |v_{\cH}|$.
    Every honest party has value $v_{\cH}[k]$ at coordinate $k$.
    By \Cref{lem:pc:qc1-threshold}, \qcone contains at least $t$ honest
    votes, so $v_{\cH}[k]$ has threshold support in \qcone.
    It is the unique threshold-supported value at coordinate $k$ and
    ${\sf QC1Certify}$ appends it.
    This holds for every $k\le |v_{\cH}|$, hence
    $v_{\cH}\preceq\qcone.\x$.
\end{proof}
\fi

\begin{theorem}\label{thm:pc:1}
    \pc Protocol (\Cref{alg:pc}) solves \pc (\Cref{def:pc}). 
\end{theorem}

\ifcameraready
\begin{proof}
The Prefix property follows from quorum intersection and
\Cref{lem:pc:1}; Validity follows from \Cref{lem:pc:2}; and eventual
delivery of the $n-f$ honest votes in each phase gives Termination.
Full details appear in \fullversionref{sec:pc:simple:analysis}.
\end{proof}
\else
\begin{proof}
    {\em Prefix.}
    For any $\qcthree,\qcthree'$, by quorum intersection, at least one honest party's $\votethree=(\xmcp, *, *)$ is in both $\qcthree$ and $\qcthree'$.
    Since $\qcthree.\xmcpmcp:=\mcp{\{\xmcp \in \qcthree\}}\preceq \xmcp$ and $\qcthree'.\xmcpmce:=\mce{\{\xmcp \in \qcthree'\}}\succeq \xmcp$ (note that by~\Cref{lem:pc:1}, all vectors in $\qcthree'$ are consistent and, hence, $\sf mce$ is well-defined), we conclude that
    $\qcthree.\xmcpmcp \preceq \qcthree'.\xmcpmce$ for any $\qcthree,\qcthree'$.
    
    According to the protocol, for any $i,j\in \cH$, $v_{i}^{\sf low}=\qcthree.\xmcpmcp$ and $v_{j}^{\sf high}=\qcthree'.\xmcpmce$ for some $\qcthree,\qcthree'$. 
    Therefore $v_{i}^{\sf low}\preceq v_{j}^{\sf high}$ for any $i,j\in \cH$.
    
    {\em Termination.}
    Function {\sf QC1Certify} performs a finite scan over the $L$ coordinates.
    At each coordinate, the threshold-supported value is unique by
    \Cref{lem:pc:qc1-threshold}; if no such value exists, the function returns the prefix computed so far.
    Function {\sf QC2Certify} terminates because $\mcp{\cdot}$ is well-defined on every finite set of vectors.
    Function {\sf QC3Certify} also terminates: the operation $\mce{\cdot}$ is well-defined since, by~\Cref{lem:pc:1}, we have $\qctwo.\xmcp \sim \qctwo'.\xmcp$ for any $\qctwo, \qctwo'$, ensuring that a common extension always exists.

    All honest parties eventually receive votes from every other honest party, and therefore each can compute the required functions and continue to vote. 
    Thus, the protocol terminates.

    {\em Validity.} 
    Let $v_{\cH}=\mcp{\{v_{h}^{\sf in}\}_{h\in \cH}}$ denote the maximum common prefix of all honest parties’ inputs.
    By~\Cref{lem:pc:2}, $v_{\cH} \preceq \qcone.\x$ for any \qcone. 
    
    By definition, $\qctwo.\xmcp= \mcp{\{\x \in \qctwo\}}$. 
    Since every $\qctwo$ only contains the $\x$-values certified by $\qcone$, and $v_{\cH} \preceq \qcone.\x$ for any \qcone, we obtain that $v_{\cH} \preceq \qctwo.\xmcp$ for any \qctwo.
    An identical argument applies to any $\qcthree$, proving that $v_{\cH} \preceq \qcthree.\xmcpmcp$.

    According to the protocol, for any $i\in\cH$, $v_{i}^{\sf low}=\qcthree.\xmcpmcp$ for some $\qcthree$. 
    Therefore $v_{\cH}\preceq v_{i}^{\sf low}$ for any $i\in \cH$. 

\end{proof}
\fi

\begin{theorem}\label{thm:pc:complexity}
    \pc Protocol (\Cref{alg:pc}) has a round complexity of $3$ rounds, message complexity of $O(n^2)$, and communication complexity of $O((\const L + \sigsize) n^4)$, where $\const$ is the size of each vector element and $\sigsize$ is the size of signature. 
\end{theorem}

\ifcameraready
\begin{proof}
Deferred to \fullversionref{sec:pc:simple:analysis}.
\end{proof}
\else
\begin{proof}
    From the protocol specification, all parties exchange messages in an all-to-all manner for three rounds. Thus, the round complexity is 3, and the message complexity of $O(n^2)$. 

    We now analyze the communication complexity. 
    A \voteone message has size $O(\const L + \sigsize)$, a \qcone has size $O((\const L + \sigsize) n)$, a \votetwo message has size $O((\const L + \sigsize) n)$, a \qctwo has size $O((\const L + \sigsize) n^2)$, a \votethree message has size $O((\const L + \sigsize) n^2)$, and a \qcthree has size $O((\const L + \sigsize) n^3)$. 
    Since each round involves an all-to-all exchange of messages, the total communication cost of the protocol is $O((\const L + \sigsize) n^2) \cdot O(n^2)= O((\const L + \sigsize) n^4)$. 
\end{proof}
\fi


\section{Lower Bound}
\label{sec:lower-bound}

\newcommand{\vlow}{v^{\sf low}}
\newcommand{\vhigh}{v^{\sf high}}
\newcommand{\alg}{\mathcal{A}}
\newcommand{\dhat}{\hat\delta}
\newcommand{\rzero}{\rho_0}
\newcommand{\rone}{\rho_1}
\newcommand{\vzero}{v_0}
\newcommand{\vone}{v_1}

\begin{theorem} \label{thm:lower-bound}
    It is impossible to solve \pc in an asynchronous network with latency $2\delta$ when $f \ge \frac{n}{4}$ parties can be Byzantine.
\end{theorem}

The rest of the section presents a proof of \Cref{thm:lower-bound} by contradiction.

\subsection{Notation}

Suppose there exists an algorithm $\alg$ solving \pc with the latency of at most $2\delta$ in asynchrony with $n \le 4f$ parties.
Let $\Pi$ be the set of all parties, $\abs{\Pi} = n \le 4f$.
We split $\Pi$ into 4 disjoint sets: $A$, $B$, $C$, and $D$ such that $\abs{A} = \abs{B} = \abs{C} = f$ and $\abs{D} = n-3f \le f$.

We will consider a number of executions of $\alg$ in order to arrive at a contradiction.

\paragraph{Executions and notations.}

Let $\mathcal V$ be the set of possible input vectors.
An \emph{initial configuration} is a function $I: \Pi \to \mathcal V$ mapping each party $i$ to its input vector $v^{\sf in}_i$.
An \emph{execution} of $\alg$ is a tuple $\rho = (I, F, S)$, where $F \subseteq \Pi$ is the set of Byzantine parties ($\abs{F} \le f$),
and $S$ is a totally ordered sequence of \emph{send}, \emph{receive}, and \emph{output} events,
each tagged with an absolute time according to a global clock inaccessible to the protocol.\footnote{As is standard in lower-bound arguments, we can assume $\alg$ is deterministic, considering any fixed set of coin values for randomized protocols.}

If $i$ is honest in $\rho = (I, F, S)$ (i.e., $i \notin F$), let $(\vlow_i(\rho), \vhigh_i(\rho))$ be the output pair of $i$ in $\rho$.
For a set $G \subseteq \Pi \setminus F$ of honest parties in $\rho$, we write $\vlow_G(\rho) = v$ (resp., $\vhigh_G(\rho) = v$) to denote that $\forall i \in G:\ \vlow_i(\rho) = v$ (resp., $\vhigh_i(\rho) = v$).

\paragraph{Similar executions.}

Given an execution $\rho = (I, F, S)$ and an honest party $i \notin F$, the \emph{output view} of $i$ in $\rho$ is its input together with all messages delivered to $i$
\emph{up to the moment when $i$ outputs $(\vlow_i, \vhigh_i)$}, including their delivery times.
For two executions $\rho,\rho'$ and a party $i$ that is honest in both, we write $\rho \stackrel{i}\sim \rho'$ if 
the output views of party $i$ in $\rho$ and $\rho'$ are the same.
Note that the messages delivered by $i$ \emph{after} it outputs could differ in $\rho$ and $\rho'$.
For a set $G$ of parties that are honest in both executions, we write $\rho \stackrel{G}\sim \rho'$ as a shorthand for $\forall i \in G:\ \rho \stackrel{i}\sim \rho'$. 

\begin{observation} \label{obs:lb:similar-exec}
    If $\rho \stackrel{G}\sim \rho'$, then $\vhigh_G(\rho) = \vhigh_G(\rho')$ and $\vlow_G(\rho) = \vlow_G(\rho')$.
\end{observation}

\paragraph{Asynchronous latency}

For an execution $\rho = (I,F,S)$, define $\delta(\rho)$ as the longest time it takes to deliver a message in $\rho$ from one honest party to another.
For a \pc algorithm $\alg$, we say that its latency is $k\delta$ if the longest time across all valid executions $\rho$ for all honest parties to output $(\vlow_i, \vhigh_i)$ in $\rho$ is at most $k\delta(\rho)$.

\begin{figure*}[t]
\centering

\newcommand{\tobiz}[1]{}

\begin{subfigure}[t]{\lbfigsize}
    \centering
    \ExecDiagram{
        \MarkByz{\yC}
    
        \draw[msgone] (0,\yA) -- (1,\yB);
        \tobiz{\draw[msgone] (0,\yA) -- (1,\yC);}
        \draw[msgone] (0,\yA) -- (1,\yD);
        \draw[msgone] (0,\yB) -- (1,\yA);
        \tobiz{\draw[msgone] (0,\yB) -- (1,\yC);}
        \draw[msgone] (0,\yB) -- (1,\yD);
        \draw[msgone] (0,\yD) -- (1,\yA);
        \draw[msgone] (0,\yD) -- (1,\yB);
        \tobiz{\draw[msgone] (0,\yD) -- (1,\yC);}
    
        \draw[msgone] (1,\yA) -- (2,\yB);
        \tobiz{\draw[msgone] (1,\yA) -- (2,\yC);}
        \draw[msgone] (1,\yA) -- (2,\yD);
        \draw[msgone] (1,\yB) -- (2,\yA);
        \tobiz{\draw[msgone] (1,\yB) -- (2,\yC);}
        \draw[msgone] (1,\yB) -- (2,\yD);
        \draw[msgone] (1,\yD) -- (2,\yA);
        \draw[msgone] (1,\yD) -- (2,\yB);
        \tobiz{\draw[msgone] (1,\yD) -- (2,\yC);}
    
        \node[decision] (decB) at (2,\yB) {};
        \node[anchor=west,font=\scriptsize] at ($(decB)+(0.12,0)$) {$\vlow_B = \vzero$};
    }
    \caption{Execution $\rzero$: $A(\vone), B(\vzero), D(\vzero)$ are honest; $C$ Byzantine and silent.}
    \label{subfig:lb:rho0}
\end{subfigure}
\hfill
\begin{subfigure}[t]{\lbfigsize}
    \centering
    \ExecDiagram{
        \begin{scope}[shift={(2.4,0.1)}]
            \draw[msgone] (0,0.25) -- (0.4,0.25)
              node[anchor=west,font=\scriptsize,blue] {same as $\rzero$};
            \draw[msgtwo] (0,0) -- (0.4,0)
              node[anchor=west,font=\scriptsize,green!55!black] {diff.\ from $\rzero$};
        \end{scope}
    
        \MarkByz{\yB}
    
        \tobiz{\draw[msgone] (0,\yA) -- (1,\yB);}
        \draw[msgone] (0,\yA) -- (1,\yC);
        \draw[msgone] (0,\yA) -- (1,\yD);
        \draw[msgone] (0,\yB) -- (1,\yD);
        \draw[msgtwo] (0,\yC) -- (1,\yA);
        \tobiz{\draw[msgtwo] (0,\yC) -- (1,\yB);}
        \draw[msgtwo] (0,\yC) -- (1.2,\yD);
        \draw[msgone] (0,\yD) -- (1,\yA);
        \tobiz{\draw[msgone] (0,\yD) -- (1,\yB);}
        \draw[msgone] (0,\yD) -- (1,\yC);
    
        \tobiz{\draw[msgtwo] (1,\yA) -- (2,\yB);}
        \draw[msgtwo] (1,\yA) -- (2,\yC);
        \draw[msgtwo] (1,\yA) -- (2,\yD);
        \draw[msgtwo] (1,\yC) -- (2,\yA);
        \tobiz{\draw[msgtwo] (1,\yC) -- (2,\yB);}
        \draw[msgtwo] (1,\yC) -- (2.2,\yD);
        \draw[msgone] (1,\yD) -- (2,\yA);
        \tobiz{\draw[msgone] (1,\yD) -- (2,\yB);}
        \draw[msgone] (1,\yD) -- (2,\yC);
        \draw[msgtwo] (1.2,\yD) -- (2.2,\yA);
        \tobiz{\draw[msgtwo] (1.2,\yD) -- (2.2,\yB);}
        \draw[msgtwo] (1.2,\yD) -- (2.2,\yC);
    
        \node[decision] (decC) at (2.4,\yC) {};
        \node[anchor=west,font=\scriptsize] at ($(decC)+(0.12,0)$) {$\vhigh_C = \vone$};
    }
    
    \caption{Execution $\rone$: $A(\vone), C(\vone), D(\vzero)$ are honest; $B$ is Byzantine and only sends messages to $D$; messages $C\to D$ are slowed down.}
    \label{subfig:lb:rho1}
\end{subfigure}
\hfill
\begin{subfigure}[t]{\lbfigsize}
    \centering
    \ExecDiagram{
    
        \begin{scope}[shift={(2.4,0.1)}]
            \draw[msgone] (0,0.5) -- (0.4,0.5)
              node[anchor=west,font=\scriptsize,blue] {same as $\rzero$};
            \draw[msgtwo] (0,0.25) -- (0.4,0.25)
              node[anchor=west,font=\scriptsize,green!55!black] {same as $\rone$};
            \draw[msgthr] (0,0) -- (0.4,0)
              node[anchor=west,font=\scriptsize,red!55!black] {diff.\ from $\rzero$ \& $\rone$};
        \end{scope}
    
        \MarkByz{\yA}
    
        \draw[msgone] (0,\yA) -- (1,\yB);
        \draw[msgone] (0,\yA) -- (1,\yC);
        \draw[msgone] (0,\yA) -- (1,\yD);
        \tobiz{\draw[msgone] (0,\yB) -- (1,\yA);}
        \draw[msgthr] (0,\yB) -- (3,\yC);
        \draw[msgone] (0,\yB) -- (1,\yD);
        \draw[msgone] (1,\yB) -- (2,\yD);
        \tobiz{\draw[msgtwo] (0,\yC) -- (1,\yA);}
        \draw[msgthr] (0,\yC) -- (3,\yB);
        \draw[msgtwo] (0,\yC) -- (1.2,\yD);
        \tobiz{\draw[msgone] (0,\yD) -- (1,\yA);}
        \draw[msgone] (0,\yD) -- (1,\yB);
        \draw[msgone] (0,\yD) -- (1,\yC);
    
        \draw[msgone] (1,\yA) -- (2,\yB);
        \draw[msgtwo] (1,\yA) -- (2,\yC);
        \draw[msgtwo] (1,\yA) -- (2,\yD);
        \tobiz{\draw[msgone] (1,\yB) -- (2,\yA);}
        \draw[msgthr] (1,\yB) -- (3,\yC);
        \draw[msgone] (1,\yB) -- (2,\yD);
        \tobiz{\draw[msgtwo] (1,\yC) -- (2,\yA);}
        \draw[msgthr] (1,\yC) -- (3,\yB);
        \draw[msgtwo] (1,\yC) -- (2.2,\yD);
        \tobiz{\draw[msgone] (1,\yD) -- (2,\yA);}
        \draw[msgone] (1,\yD) -- (2,\yB);
        \draw[msgone] (1,\yD) -- (2,\yC);
        \tobiz{\draw[msgtwo] (1.2,\yD) -- (2.2,\yA);}
        \draw[msgtwo] (1.2,\yD) -- (2.2,\yB);
        \draw[msgtwo] (1.2,\yD) -- (2.2,\yC);
    
        \node[decision] (decB) at (2,\yB) {};
        \node[decision] (decC) at (2.4,\yC) {};
    }
    
    \caption{Execution $\rho^*$: 
    $A$ is Byzantine, simulating $\rzero$ toward $B$ and $\rone$ toward $C$;
    $B$'s view up to $2\delta(\rzero) = 2\dhat$ matches $\rzero$ and $C$'s view up to $2\delta(\rone) = 2.4\dhat$ matches $\rone$.
    }
    \label{subfig:lb:rho-star}
\end{subfigure}

\vspace{2pt}

\caption{Executions $\rzero$, $\rone$, and $\rho^*$ used in the lower-bound proof. For clarity, messages to Byzantine parties and messages sent at or after $2\dhat$ are omitted from the picture.}

\label{fig:lb:fig1}
\end{figure*}

\begin{figure*}[t]
\centering

\newcommand{\tobiz}[1]{}
\newcommand{\tobizimportant}[1]{#1}

\begin{subfigure}[t]{\lbfigsize}
    \centering
    \ExecDiagram{
        \begin{scope}[shift={(2.4,0.1)}]
            \draw[msgone] (0,0.25) -- (0.4,0.25)
              node[anchor=west,font=\scriptsize,blue] {same as $\rone''$};
            \draw[msgone] (0,0) -- (0.4,0)
              node[anchor=west,font=\scriptsize,green!55!black] {diff.\ from $\rone''$};
        \end{scope}
    
        \MarkByz{\yB}
    
        \tobiz{\draw[msgone] (0,\yA) -- (1,\yB);}
        \draw[msgone] (0,\yA) -- (1,\yC);
        \draw[msgone] (0,\yA) -- (1,\yD);
        \draw[msgone] (0,\yB) -- (1,\yD);
        \draw[msgone] (0,\yC) -- (1,\yA);
        \tobiz{\draw[msgone] (0,\yC) -- (1,\yB);}
        \draw[msgtwo] (0,\yC) -- (1.2,\yD);
        \draw[msgone] (0,\yD) -- (1,\yA);
        \tobiz{\draw[msgone] (0,\yD) -- (1,\yB);}
        \draw[msgone] (0,\yD) -- (1,\yC);
    
        \tobiz{\draw[msgone] (1,\yA) -- (2,\yB);}
        \draw[msgone] (1,\yA) -- (2,\yC);
        \draw[msgone] (1,\yA) -- (2,\yD);
        \draw[msgone] (1,\yC) -- (2,\yA);
        \tobiz{\draw[msgone] (1,\yC) -- (2,\yB);}
        \draw[msgtwo] (1,\yC) -- (2.2,\yD);
        \draw[msgtwo] (1,\yD) -- (2,\yA);
        \tobiz{\draw[msgtwo] (1,\yD) -- (2,\yB);}
        \draw[msgtwo] (1,\yD) -- (2,\yC);
        \draw[msgtwo] (1.2,\yD) -- (2.2,\yA);
        \tobiz{\draw[msgtwo] (1.2,\yD) -- (2.2,\yB);}
        \draw[msgtwo] (1.2,\yD) -- (2.2,\yC);
    
        \node[decision] (decC) at (2.4,\yC) {};
        \node[anchor=west,font=\scriptsize] at ($(decC)+(0.12,0)$) {$\vhigh_C = \vone$};
    }
    
    \caption{Execution $\rone$: $A(\vone), C(\vone), D(\vzero)$ are honest; $B$ is Byzantine and only sends messages to $D$; messages $C\to D$ are slowed down. $\vhigh_C(\rone) = \vone$ because $\rone \stackrel{C}\sim \rone'$.}
    \label{subfig:lb:fig2:rho1}
\end{subfigure}
\hfill
\begin{subfigure}[t]{\lbfigsize}
    \centering
    \ExecDiagram{
        \begin{scope}[shift={(2.4,0.1)}]
            \draw[msgone] (0,0.5) -- (0.4,0.5)
              node[anchor=west,font=\scriptsize,blue] {same as $\rone''$};
            \draw[msgtwo] (0,0.25) -- (0.4,0.25)
              node[anchor=west,font=\scriptsize,green!55!black] {same as $\rone$};
            \draw[msgthr] (0,0) -- (0.4,0)
              node[anchor=west,font=\scriptsize,red!55!black] {diff.\ from $\rone \& \rone''$};
        \end{scope}
    
        \MarkByz{\yD}
    
        \draw[msgone] (0,\yA) -- (1,\yB);
        \draw[msgone] (0,\yA) -- (1,\yC);
        \tobiz{\draw[msgone] (0,\yA) -- (1,\yD);}
        \draw[msgthr] (0,\yB) -- (3,\yA);
        \draw[msgthr] (0,\yB) -- (3,\yC);
        \tobizimportant{\draw[msgone] (0,\yB) -- (1,\yD);}
        \draw[msgone] (0,\yC) -- (1,\yA);
        \draw[msgone] (0,\yC) -- (1,\yB);
        \tobiz{\draw[msgone] (0,\yC) -- (1.2,\yD);}
        \draw[msgone] (0,\yD) -- (1,\yA);
        \draw[msgone] (0,\yD) -- (1,\yB);
        \draw[msgone] (0,\yD) -- (1,\yC);
    
        \draw[msgone] (1,\yA) -- (2,\yB);
        \draw[msgone] (1,\yA) -- (2,\yC);
        \tobiz{\draw[msgone] (1,\yA) -- (2,\yD);}
        \draw[msgthr] (1,\yB) -- (3,\yA);
        \draw[msgthr] (1,\yB) -- (3,\yC);
        \draw[msgone] (1,\yC) -- (2,\yA);
        \draw[msgone] (1,\yC) -- (2,\yB);
        \tobiz{\draw[msgone] (1,\yC) -- (2.2,\yD);}
        \draw[msgone] (1,\yD) -- (2,\yA);
        \draw[msgtwo] (1,\yD) -- (2,\yB);
        \draw[msgtwo] (1,\yD) -- (2,\yC);
        \draw[msgtwo] (1.2,\yD) -- (2.2,\yB);
        \draw[msgtwo] (1.2,\yD) -- (2.2,\yC);
    
        \node[decision] (decA) at (2,\yA) {};
        \node[anchor=west,font=\scriptsize] at ($(decA)+(0.12,0)$) {$\vlow_A = \vone$};
        \node[decision] (decC) at (2.4,\yC) {};
    }
    
    \caption{Execution $\rone'$: $A(\vone)$, $B(\vzero)$, and $C(\vone)$ are honest; $D$ is Byzantine and simulates $\rone''$ in its messages to $A$. 
    $\vlow_A(\rone') = \vone$ because $\rone' \stackrel{A}\sim \rone''$ $\Rightarrow$ $\vhigh_C(\rone') = \vone$.}
    \label{subfig:lb:fig2:rho1p}
\end{subfigure}
\hfill
\begin{subfigure}[t]{\lbfigsize}
    \centering
    \ExecDiagram{
    
        \MarkByz{\yB}
    
        \tobiz{\draw[msgone] (0,\yA) -- (1,\yB);}
        \draw[msgone] (0,\yA) -- (1,\yC);
        \draw[msgone] (0,\yA) -- (1,\yD);
        \draw[msgone] (0,\yC) -- (1,\yA);
        \tobiz{\draw[msgone] (0,\yC) -- (1,\yB);}
        \draw[msgone] (0,\yC) -- (1,\yD);
        \draw[msgone] (0,\yD) -- (1,\yA);
        \tobiz{\draw[msgone] (0,\yD) -- (1,\yB);}
        \draw[msgone] (0,\yD) -- (1,\yC);
    
        \tobiz{\draw[msgone] (1,\yA) -- (2,\yB);}
        \draw[msgone] (1,\yA) -- (2,\yC);
        \draw[msgone] (1,\yA) -- (2,\yD);
        \draw[msgone] (1,\yC) -- (2,\yA);
        \tobiz{\draw[msgone] (1,\yC) -- (2,\yB);}
        \draw[msgone] (1,\yC) -- (2,\yD);
        \draw[msgone] (1,\yD) -- (2,\yA);
        \tobiz{\draw[msgone] (1,\yD) -- (2,\yB);}
        \draw[msgone] (1,\yD) -- (2,\yC);
    
        \node[decision] (decA) at (2,\yA) {};
        \node[anchor=west,font=\scriptsize] at ($(decA)+(0.12,0)$) {$\vlow_A = \vone$};
    }
    
    \caption{Execution $\rone''$: $A(\vone)$, $C(\vone)$, and $D(\vzero)$ are honest; $B$ is Byzantine and silent. $\vlow_A(\rone'') = \vone$ by \Cref{lem:lb:validity}.}
    \label{subfig:lb:fig2:rho1pp}
\end{subfigure}

\vspace{2pt}

\caption{Executions $\rone$, $\rone'$, and $\rone''$ from the proof of \Cref{lem:lb:rho-one}. For clarity, messages to Byzantine parties and messages sent at or after $2\dhat$ are omitted from the picture.}

\label{fig:lb:fig2}
\end{figure*}

\subsection{Proof}

Consider 2 conflicting valid length-$L$ input vectors: $\vzero$ and $\vone$.
We will build 2 executions: $\rzero$ and $\rone$ such that $\vlow_B(\rzero) = \vzero$  and $\vhigh_C(\rone) = \vone$.
We will then construct a single execution $\rho^*$ where $\rho^* \stackrel{B}\sim \rzero$ and $\rho^* \stackrel{C}\sim \rone$.
By \Cref{obs:lb:similar-exec}, $\vlow_B(\rho^*) = \vzero$ and $\vhigh_C(\rho^*) = \vone$.
This violates the Prefix property of \pc---A contradiction with $\alg$ solving \pc.

Let us start with an auxiliary lemma.

\begin{lemma} \label{lem:lb:validity}
    Let $\rho$ be any valid execution of $\alg$ where all honest parties have input in $\{\vzero, \vone\}$.
    Let $b \in \{0, 1\}$, $H_b$ be the set of honest parties with input $v_b$, $H_{1-b}$ be the set of honest parties with input $v_{1-b}$, and $F$ be the set of Byzantine parties in $\rho$.
    If $\abs{H_{1-b}} \le f$ and the parties in $F$ are silent, then $\vlow_{H_b}(\rho) = v_b$.
\end{lemma}
\begin{proof}
    Intuitively, from the perspective of any party in $H_b$, the parties in $H_{1-b}$ could be Byzantine,
    while the parties in $F$ could be honest with input $v_b$ whose messages are delayed. Hence, it could be that all honest parties have input $v_b$.

    More formally, let $\rho'$ be an execution identical to $\rho$, except:
    \begin{itemize}
        \item Parties in $H_{1-b}$ are Byzantine, but they act identically to $\rho$; all messages between $H_b$ and $H_{1-b}$ are delivered in the same order $\rho$.

        \item Parties in $F$ are honest with input $v_b$, but all their messages are delayed until time $3\delta(\rho)$.
    \end{itemize}

    In $\rho$, $H_b$ must output by the time $2\delta(\rho)$.
    Note, however, that for any party in $H_b$, $\rho'$ is identical to $\rho$ up to the time $3\delta(\rho)$.
    Hence, in $\rho'$, parties in $H_b$ must also output by the time $2\delta(\rho)$ and $\rho \stackrel{H_b}\sim \rho'$.
    Moreover, since all honest parties have input $v_b$ in $\rho'$, by Validity, $\vlow_{H_b}(\rho') = v_b$.
    By \Cref{obs:lb:similar-exec}, $\vlow_{H_b}(\rho) = \vlow_{H_b}(\rho') = v_b$.
\end{proof}

Let $\dhat$ be an arbitrary time interval.
From here on out, unless otherwise specified, in the executions that we are constructing, all parties start at the same time, and all messages take exactly $\dhat$ time to be delivered, with ties in the delivery order broken deterministically (e.g., messages from $A$ first, then $B$, $C$, and $D$).
Moreover, groups $A$ and $C$ will always have input $\vone$, and groups $B$ and $D$---$\vzero$.

Now, we can construct our 2 executions $\rzero$ and $\rone$  (illustrated in \Cref{subfig:lb:rho0,subfig:lb:rho1}) that we will later combine to arrive at an execution $\rho^*$ (\Cref{subfig:lb:rho-star}) where the Prefix property is violated:
\begin{itemize}
    \item $\rzero$ (\Cref{subfig:lb:rho0}): $A(\vone)$, $B(\vzero)$, and $D(\vzero)$ are honest; $C$ is Byzantine and silent.

    \item $\rone$ (\Cref{subfig:lb:rho1}): $A(\vone)$, $C(\vone)$, and $D(\vzero)$ are honest; $B$ is Byzantine.
    Parties in $B$ send the same messages to $D$ as in $\rzero$, but omit all other messages.
    Messages from $C \to D$ are slowed down and take $1.2\dhat$ instead of $1\dhat$ to be delivered.
\end{itemize}

\begin{lemma}
    $\vlow_B(\rzero) = \vzero$
\end{lemma}
\begin{proof}
    Follows directly from~\Cref{lem:lb:validity}, where $b=1$, $H_b = B \cup D$, $H_{1-b} = A$, and $F = C$.
\end{proof}

\begin{lemma} \label{lem:lb:rho-one}
    $\vhigh_C(\rone) = \vone$
\end{lemma}
\begin{proof}
    Intuitively, from the perspective of parties in $C$, $D$ could be Byzantine and not have forwarded the information about group $B$ to group $A$. In that case, $A$ could have output $\vlow = \vone$, so to respect the Prefix property of \pc, $C$ must output $\vhigh = \vone$.

    More formally, let $\rone'$ (\Cref{subfig:lb:fig2:rho1p}) be an execution identical to $\rone$ (\Cref{subfig:lb:fig2:rho1}, except:
    \begin{itemize}
        \item $D$ is Byzantine and acts identically to $\rone$ toward parties in $C$, but pretends not to have seen any messages from $B$ in their messages toward $A$.
        \item $B(1)$ is honest, but its messages to $\Pi \setminus D$ are delayed until $3\dhat$; all other messages are delivered in the same order as in $\rone$. 
    \end{itemize}
    In $\rone$, $C$ must output by the time $2\delta(\rone) = 2.4\dhat$.
    Moreover, the set of messages received by $C$ in $\rone'$ is identical to $\rone$ up to the time $3\dhat$ when it receives the delayed messages from $B$.
    Hence, $C$ must output by the time $2.4\dhat$ in $\rone'$ as well and $\rone' \stackrel{C}\sim \rone$.
    
    Finally, we demonstrate that $\vlow_A(\rone') = \vone$ and, thus, by the Prefix property of \pc, $\vhigh_C(\rone') = \vone$.
    To this end, consider execution $\rone''$ (\Cref{subfig:lb:fig2:rho1pp}) where $A(\vone)$, $C(\vone)$, and $D(\vzero)$ are honest; $B$ is Byzantine and silent.
    By~\cref{lem:lb:validity}, $\vlow_A(\rone'') = \vone$.
    The set of messages received by $A$ in $\rone'$ and $\rone''$ before the time $2\dhat$ is identical and thus $\rone' \stackrel{A}\sim \rone''$ and $\vlow_A(\rone') = \vlow_A(\rone'') = \vone$.
    Hence, $\vhigh_C(\rone') = \vone$.
    Since $\rone \stackrel{C}\sim \rone'$, $\vhigh_C(\rone) = \vhigh_C(\rone') = \vone$.
\end{proof}

Finally, let us construct the execution $\rho^*$ (\Cref{subfig:lb:rho-star}) by combining $\rzero$ and $\rone$ in a way that causes the disagreement.
In $\rho^*$, $A$ is Byzantine. It mimics $\rzero$ in its messages toward parties in $B$ and $\rone$ in its messages toward parties in $C$ and $D$.
Similarly to $\rone$, messages from $C$ to $D$ are slowed down and take $1.2\dhat$ to be delivered.
Messages between $B$ and $C$ are delayed to $3\dhat$.
As illustrated in \Cref{fig:lb:fig1}, messages delivered by $B$ up to the time $2\dhat$ are identical to $\rzero$ and since $B$ outputs at time $2\dhat$ in $\rzero$, it must also output at time $2\dhat$ in $\rho^*$ and $\rho^* \stackrel{B}\sim \rzero$
Similarly, $\rho^* \stackrel{C}\sim \rone$.
Hence, $\vlow_B(\rho^*) = \vzero$ and $\vhigh_C(\rho^*) = \vone$, violating the Prefix property of \pc---A contradiction.

\section{\pcver}
\label{sec:pcver}

An outer protocol such as \pcstr must safely process \pc outputs presented
by potentially Byzantine parties.
\pcver therefore attaches publicly checkable proofs to the low and high
outputs.
The \pc protocol of \Cref{sec:pc} realizes \pcver without additional
messages or rounds by exposing its existing \qcthree as the output proof.

\begin{definition}[\pcver]\label{def:pcver}
\pcver satisfies \pc, and each party $i$ additionally outputs proofs
$\pi_i^{\sf low}$ and $\pi_i^{\sf high}$ for its low and high outputs.
For every instance identifier $\id$, public predicates
$\mathcal F_{\id}^{\sf low}$ and
$\mathcal F_{\id}^{\sf high}$ satisfy:
\begin{itemize}
    \item \textbf{Proof Completeness.}
    For every honest party $i$, its output proofs are accepted:
    \[
    \mathcal F_{\id}^{\sf low}(v_i^{\sf low},\pi_i^{\sf low})
    =
    \mathcal F_{\id}^{\sf high}(v_i^{\sf high},\pi_i^{\sf high})
    =
    {\sf true}.
    \]

    \item \textbf{Proof Soundness.}
    If $(v^{\sf low},\pi^{\sf low})$ and
    $(v^{\sf high},\pi^{\sf high})$ are publicly accepted for the same
    instance, then
    \[
    v^{\sf low}\preceq v^{\sf high}
    \qquad\text{and}\qquad
    \mcp{\{v_h^{\sf in}\}_{h\in\cH}}\preceq v^{\sf low}.
    \]
\end{itemize}
\end{definition}

\paragraph{Construction.}
In the protocol from \Cref{alg:pc}, each party $i$ sets
$\pi_i^{\sf low}=\pi_i^{\sf high}=\qcthree$.
Predicate ${\sf ValidQC3}(\id,\qcthree)$ recursively checks that:
(i) every QC contains $n-f$ valid messages from distinct senders for
instance $\id$;
(ii) all signatures and nested certificates are valid; and
(iii) every claimed value equals the output of its certification
procedure.
Define
\[
\begin{aligned}
\mathcal F_{\id}^{\sf low}(v,\qcthree)
&\Longleftrightarrow
{\sf ValidQC3}(\id,\qcthree)
\ \wedge\ (v,*)={\sf QC3Certify}(\qcthree),\\
\mathcal F_{\id}^{\sf high}(v,\qcthree)
&\Longleftrightarrow
{\sf ValidQC3}(\id,\qcthree)
\ \wedge\ (*,v)={\sf QC3Certify}(\qcthree).
\end{aligned}
\]

\begin{theorem}\label{thm:pc:verifiable}
\Cref{alg:pc}, augmented with its nested \qcthree output proof, solves
\pcver with the same round, message, and communication complexities as
\pc.
\end{theorem}
\begin{proof}
Every honest output is computed from a valid \qcthree, so its proof is
accepted by the corresponding predicate, establishing Proof Completeness.
Recursive validation ensures that every accepted certificate satisfies
the conditions used in the analysis of \Cref{alg:pc}.
The arguments for Prefix and Validity in \Cref{thm:pc:1} therefore apply
to all publicly accepted outputs.
No additional protocol messages or certificate data are introduced.
\end{proof}

\section{\pcstr}
\label{sec:pcstr-leaderless}

In this section, we present a {\em leaderless} (\Cref{def:leaderless}) construction of \pcstr that makes black-box use of \pcver.
For the case considered here, the adversary may suspend one process per
round while up to $f-1$ other processes are Byzantine.
Our construction has worst-case message complexity $O(n^3)$ and communication complexity $O(n^4)$ (\fullversionref{thm:pcstr:complexity}).
As noted in \Cref{sec:def:pcstr:relation}, any \pcstr protocol immediately
yields Binary Consensus.
Consequently, our construction yields leaderless Binary Consensus with
worst-case message complexity $O(n^3)$, improving over the prior
$O(n^4)$ bound of~\cite{antoniadis2021leaderless}.

\begin{algorithm}[t!]
\caption{\pcstr Protocol using Leaderless Validated Consensus Protocol}
\label{alg:pcstr-leaderless:simple}
\centering
\begin{algorithmic}[1]

\Function{\sf Predicate}{$v,\pi$}\label{line:pcstr-simple:predicate}
    \State \Return $\cF^{\sf high}(v,\pi)$
\EndFunction

\Statex

\Upon{{\bf input} $v_{i}^{\sf in}$}
    \State run protocol $\Pi^{\sf VPC}$ with input $v_{i}^{\sf in}$
\EndUpon

\Statex

\Upon{$\Pi^{\sf VPC}$ outputs $(v^{\sf low}_{i},\pi^{\sf low}_{i})$}
    \State {\bf output} $v^{\sf low}_i$
    \Comment{Output low value}
\EndUpon

\Statex

\Upon{$\Pi^{\sf VPC}$ outputs $(v^{\sf high}_{i},\pi^{\sf high}_{i})$}
    \State run protocol $\Pi^{\sf LVC}$ with input $v^{\sf high}_{i}$ and proof $\pi^{\sf high}_{i}$
    \Statex \Comment{$\Pi^{\sf LVC}$ has predicate defined in line~\ref{line:pcstr-simple:predicate}}
\EndUpon

\Statex

\Upon{$\Pi^{\sf LVC}$ outputs $v^{\sf high}_{i}$}
    \State {\bf output} $v^{\sf high}_i$
    \Comment{Output high value}
\EndUpon

\end{algorithmic}
\end{algorithm}

\paragraph{Warmup: a simple construction from leaderless Validated Consensus.}
Suppose we are given a leaderless Validated Consensus protocol $\Pi^{\sf LVC}$. Then we obtain a simple construction of \pcstr by combining a \pcver protocol $\Pi^{\sf VPC}$ with $\Pi^{\sf LVC}$, as shown in \Cref{alg:pcstr-leaderless:simple}.
At a high level, parties first invoke $\Pi^{\sf VPC}$ on their input and immediately output the resulting low value. To agree on the high value, they invoke $\Pi^{\sf LVC}$ on the high-value candidate produced by \pcver together with its proof, and output the decision of $\Pi^{\sf LVC}$ as the high value. Since \pcver produces verifiable high values, the external-validity predicate for $\Pi^{\sf LVC}$ simply checks whether the supplied proof certifies the high value.

Correctness is immediate: Agreement on the high output follows from Agreement of $\Pi^{\sf LVC}$; the Prefix property follows from the corresponding property of \pcver together with its Verifiability and the External Validity of $\Pi^{\sf LVC}$; Validity follows from Validity of \pcver; and Termination follows from Termination of both \pcver and $\Pi^{\sf LVC}$.

However, to the best of our knowledge, no such leaderless Validated Consensus protocol is known. The best existing leaderless construction incurs worst-case message (and communication) complexity $O(n^4)$ under the validity requirement that every decided value was previously proposed~\cite{antoniadis2021leaderless}.

\paragraph{Overview and intuition of our protocol.}
The protocol runs a sequence of \emph{views}. In each view $w$, parties invoke a \pcver instance $\Pi^{\sf VPC}_w$. Each instance outputs a verifiable \emph{low} value (safe to commit) and a verifiable \emph{high} value (safe to extend), each with a publicly checkable proof.

View~1 runs \pcver on the \pcstr input vector (e.g., the set of block proposals in later \smr constructions).
Its verifiable low output is immediately taken as the low output of \pcstr.
The remaining task is to agree on a verifiable high value that extends this committed low.

Later views do not add new inputs. Instead, parties exchange publicly verifiable certificates derived from earlier views and run \pcver on the resulting certificate vector.
The vector order is deterministic and cycles across views.
Once some later view commits a non-empty certificate vector, its first
entry unequal to $\hash(\bot)$ defines a publicly checkable parent pointer
to an earlier view's verifiable high.
Following parent pointers yields a unique chain back to view~1, which fixes a unique decided view-1 high output and hence Agreement.

Cyclic shifting also drives Termination in a leaderless way.
After GST, the shift ensures that eventually two honest parties occupy the first two positions of the order; thus, even if the adversary can suspend one party per round (\Cref{def:leaderless}), it cannot indefinitely prevent some view from committing a non-empty certificate vector.

\begin{algorithm*}[t!]
\caption{\pcstr Protocol for party $i$}
\label{alg:pcstr-leaderless}
\centering
\vspace{-0.6em}
\begin{algorithmic}[1]

\begin{multicols}{2}
\raggedcolumns
\setlength{\columnsep}{1.2em}

\Statex {\bf Local variables (party $i$):}
\Statex $w_i$: current view, initialized to $1$.
\Statex $\rank^{\sf SPC}_{i,w}:={\sf Shift}^{w-1}(\rank^{\sf SPC}_1)$:
public ranking for every view $w$.
\Statex $P_{i,w}[\cdot]$: length-$n$ buffer of received proposal objects for view $w>1$, initialized to all $\bot$.
\Statex $Q_{i,w}$: set of received \emptyview messages for view $w$, initialized to $\emptyset$.
\Statex $(w^{\sf high}_\star, v^{\sf high}_\star,\pi^{\sf high}_\star)$: max view $w^{\sf high}_\star$ for which $i$ holds a verifiable high output, initialized to $(0,[\,],\bot)$.

\Function{\sf High}{$\cert$}
    \If{$\cert=\cert^{\sf dir}(w,v,\pi)$
        \textbf{or} $\cert=\cert^{\sf ind}(*,(w,v,\pi),*)$}
        \State \Return $(w,v,\pi)$
    \EndIf
    \State \Return $(0,[\,],\bot)$
\EndFunction
\Function{\sf Parent}{$v$}\label{line:smr:parent}
    \For{$k=1,\dots,|v|$}
        \If{$v[k]\neq \hash(\bot)$}
            \State fetch preimage $P_k=(w,\cert)$ of $v[k]$ if needed
            \State \Return ${\sf High}(\cert)$
        \EndIf
    \EndFor
    \State \Return $(0,[\,],\bot)$
\EndFunction
\Function{\sf HasParent}{$w, v$}
    \State \Return $w=1 \vee {\sf Parent}(v)\neq (0,[\,],\bot)$
\EndFunction
\Function{\sf ValidHigh}{$w,v,\pi$}
    \State \Return $\cF^{\sf high}_w(v,\pi)\wedge{\sf HasParent}(w,v)$
\EndFunction
\Function{\sf ValidSkip}{$w,\Sigma$}
    \State \Return $\Sigma$ is a valid $f+1$-signature aggregate
    \Statex \hspace{\algorithmicindent}
        $\wedge\ \forall(w',w_h,d_h)\in\Sigma:
        w'=w\wedge 0<w_h<w$
\EndFunction
\Function{\sf ValidCert}{$w,\cert$}
    \If{$\cert=\cert^{\sf dir}(w-1,v,\pi)$}
         \State \Return ${\sf ValidHigh}(w-1,v,\pi)$
    \ElsIf{$\cert=\cert^{\sf ind}(w-1,(w_{\star},v_{\star},\pi_{\star}),\Sigma)$}
        \If{$\neg{\sf ValidSkip}(w-1,\Sigma)$} \State \Return {\sf false} \EndIf
        \State $w_h := \max\{\,w_h' : (w',w_h',*)\in\Sigma\,\}$
        \State \Return $\begin{aligned}[t]
            &w_{\star}=w_h
            \ \wedge\ (w-1,w_\star,\hash(v_\star,\pi_\star))\in\Sigma\\
            &\wedge\ {\sf ValidHigh}(w_\star,v_\star,\pi_\star)
            \end{aligned}$
    \Else \, \Return {\sf false}
    \EndIf
\EndFunction

\Upon{{\bf input} $v_{i}^{\sf in}$}
    \State run protocol $\Pi^{\sf VPC}_1$ with input $v_{i}^{\sf in}$
\EndUpon

\Upon{first invocation of {\sf RunVPC}$(w)$}
    \State $\rank^{\sf SPC}_{i,w}=(p_1,\dots,p_n):={\sf Shift}(\rank^{\sf SPC}_{i,w-1})$
    \State $v^{\sf in}_{i,w}:=[\,\hash(P_{i,w}[p_1]),\dots,\hash(P_{i,w}[p_n])\,]$
    \State run $\Pi^{\sf VPC}_w$ with input $v^{\sf in}_{i,w}$
\EndUpon

\Upon{$\Pi^{\sf VPC}_w$ outputs $(v^{\sf low}_{i,w},\pi^{\sf low}_{i,w})$}
    \State broadcast $\newcommit:=(w, v^{\sf low}_{i,w},\pi^{\sf low}_{i,w})$
\EndUpon

\Upon{$\Pi^{\sf VPC}_w$ outputs $(v^{\sf high}_{i,w},\pi^{\sf high}_{i,w})$}
    \If{${\sf HasParent}(w, v^{\sf high}_{i,w})$}
        \If{$w>w^{\sf high}_\star$} \State
            $(w^{\sf high}_\star,v^{\sf high}_\star,\pi^{\sf high}_\star)
            :=(w,v^{\sf high}_{i,w},\pi^{\sf high}_{i,w})$ \EndIf
        \State broadcast $\newview:=(w+1,\cert^{\sf dir}(w, v^{\sf high}_{i,w},\pi^{\sf high}_{i,w}))$
    \Else
        \State $d_\star:=\hash(v^{\sf high}_\star,\pi^{\sf high}_\star)$
        \State broadcast $\emptyview :=(w,(w^{\sf high}_\star, v^{\sf high}_\star,\pi^{\sf high}_\star),\sig_i((w,w^{\sf high}_\star,d_\star)))$
    \EndIf
\EndUpon

\Upon{receiving $\newview=(w,\cert)$ from party $p$}
    \If{${\sf ValidCert}(w,\cert) \wedge w\geq w_i$}
        \If{$w>w_i$}
            \Comment{Enter new view}
            \State broadcast $\newview:=(w,\cert)$
            \State $w_i := w$
            \State start $\timer(w_i)$ for $2\Delta$
        \EndIf
        \State $(w',v,\pi):={\sf High}(\cert)$
        \If{$w'>w^{\sf high}_\star$} \State
            $(w^{\sf high}_\star,v^{\sf high}_\star,\pi^{\sf high}_\star)
            :=(w',v,\pi)$ \EndIf
        \State $P_{i,w}[p]:= (w,\cert)$
        \If{$P_{i,w}[j]\neq \bot$ for $\forall j\in[n]$}
            {\sf RunVPC}$(w)$
        \EndIf
    \EndIf
\EndUpon

\Upon{receiving $\emptyview=(w,(w_h,v_h,\pi_h),\sigma)$ from party $j$}
    \State $d_h:=\hash(v_h,\pi_h)$
    \If{$w\ge w_i \wedge w>w_h \wedge
        \ver_j((w,w_h,d_h),\sigma)
        \wedge{\sf ValidHigh}(w_h,v_h,\pi_h)$}
        \State $Q_{i,w} := Q_{i,w}\cup\{(j, w_h,v_h,\pi_h,\sigma)\}$ if $(j, *, *, *, *)\not\in Q_{i,w}$
        \If{$|Q_{i,w}|=f+1$}
            \State $\Sigma :=$ aggregate signatures on
                $(w,w_h,\hash(v_h,\pi_h))$ from $Q_{i,w}$
            \State $(w_\star, v_\star,\pi_\star):=$ any triple with max $w_h$
            \State $\cert:=\cert^{\sf ind}(w, (w_\star, v_\star,\pi_\star),\Sigma)$
            \State broadcast $\newview:=(w+1,\cert)$
        \EndIf
    \EndIf
\EndUpon

\Upon{receiving $\newcommit=(w,v,\pi)$}
    \If{$\cF^{\sf low}_w(v,\pi)$}
        \State broadcast $\newcommit$ if not broadcast for $w$
        \State {\sf Commit}$(w, v,\pi)$
    \EndIf
\EndUpon

\Upon{$\timer(w)$ expires and $w=w_i$}
    \State {\sf RunVPC}$(w)$
\EndUpon

\Upon{{\sf Commit}$(w, v,\pi)$}
    \If{$w=1$}
        \If{$i$ has not output a low value}
            \State {\bf output} $(v^{\sf low}_i,\pi^{\sf low}_i):=(v,\pi)$
        \EndIf
    \ElsIf{${\sf HasParent}(w,v)$}
        \State $(w_p,v_p,\pi_p):={\sf Parent}(v)$
        \If{$w_p=1$ and $i$ has not output a high value}
            \State {\bf output} $(v^{\sf high}_i,\pi^{\sf high}_i):=(v_p,\pi_p)$
        \ElsIf{$w_p>1$}
            \State {\sf Commit}$(w_p, v_p,\pi_p)$
        \EndIf
    \EndIf
\EndUpon

\end{multicols}
\end{algorithmic}
\vspace{-0.6em}
\end{algorithm*}

\subsection{Definition}
\label{sec:pcstr-leaderless:def}

In each view $w$, parties invoke $\Pi^{\sf VPC}_w$, with public low/high predicates $\cF^{\sf low}_w(\cdot)$ and $\cF^{\sf high}_w(\cdot)$.
In addition to the usual \pcstr input, the protocol takes a public
\emph{initial ranking} $\rank^{\sf SPC}_1$, identical at all honest
parties.
It fixes the party order used by the outer protocol (e.g., \smr) to form
the view-1 \pcver input vector; all later rankings are derived
deterministically from it.

\subsubsection{Ranking.}
\label{sec:pcstr-leaderless:def:ranking}
The input ranking $\rank^{\sf SPC}_{1}$ (e.g., $(1,2,\ldots,n)$) is
used in view~1.
For every view $w\ge 2$, the ranking is updated by a fixed cyclic shift:
\[
\rank^{\sf SPC}_{i,w} := {\sf Shift}(\rank^{\sf SPC}_{i,w-1}) = (p_2,p_3,\ldots,p_n,p_1)
\]
where $\rank^{\sf SPC}_{i,w-1}=(p_1,\ldots,p_n)$.
These rankings are defined for every view independently of whether party
$i$ enters the intermediate views.
All honest parties therefore agree on $\rank^{\sf SPC}_{i,w}$ for every view $w$.

\paragraph{Instance binding.}
All protocol objects implicitly bind the \pcstr instance $\sid$, message
type, and view.
The view-$w$ \pcver identifier additionally binds
$\rank^{\sf SPC}_{i,w}$:
\[
\id_w:=\hash(\sid,w,\rank^{\sf SPC}_{i,w}).
\]
We omit these fields from the notation, and treat every broadcast as
locally processed by its sender.

\subsubsection{\pcver Input.}
In view~1, each party $i$ runs $\Pi^{\sf VPC}_1$ on its \pcstr input vector $v^{\sf in}_{i}$ (as in \Cref{def:pcstr}).
For each view $w>1$, party $p$ broadcasts a proposal object $P_{p,w}=(w,\cert)$, where $\cert$ is a certificate authorizing entry into view $w$ (defined below).
Let $P_{i,w}[p]$ be party $i$'s view of $P_{p,w}$ (or $\bot$ if missing). Party $i$ orders these by $\rank^{\sf SPC}_{i,w}=(p_1,\ldots,p_n)$ and inputs to $\Pi^{\sf VPC}_w$:
\[
v^{\sf in}_{i,w} = [\,\hash(P_{i,w}[p_1]),\ldots,\hash(P_{i,w}[p_n])\,],
\]
where $\hash(\bot)$ is a fixed distinguished value.

\subsubsection{Certificates.}
\label{sec:pcstr-leaderless:def:cert}

A certificate $\cert_{i,w}$ authorizing entry into view $w$ is of one of the following forms:
\begin{itemize}
    \item \textbf{Direct certificate:}\;
    $\mathsf{Cert}^{\sf dir}(w-1, v^{\sf high}_{w-1}, \pi^{\sf high}_{w-1})$,
    where $(v^{\sf high}_{w-1}, \pi^{\sf high}_{w-1})$ is a high output of $\Pi^{\sf VPC}_{w-1}$ that has a parent.

    \item \textbf{Indirect certificate:}\;
    $\mathsf{Cert}^{\sf ind}(w-1,(w_\star, v^{\sf high}_{w_\star}, \pi^{\sf high}_{w_\star}), \Sigma)$
    where $0<w_\star < w-1$.
    $\Sigma$ is an aggregate signature on $f+1$ skip statements of the
    form
    \[
    (w-1,w',\hash(v^{\sf high}_{w'},\pi^{\sf high}_{w'})),
    \]
    where each signer reports its latest parented high
    $(w',v^{\sf high}_{w'},\pi^{\sf high}_{w'})$ with $0<w'<w-1$.
    Define
    \[
    w_\star := \max\{\, w' : (w-1,w',*) \in \Sigma \,\}.
    \]
    The attached $(v^{\sf high}_{w_\star},\pi^{\sf high}_{w_\star})$
    must match one of the signed statements attaining this maximum and
    serves as the parent.
    Thus $\Sigma$ certifies that no view in $(w_\star,w-1]$ has a parented
    high known to the above $f+1$ parties.
\end{itemize}

\subsubsection{Parent.}
Committed values in views $w>1$ are vectors of hashes of proposal objects.
The \emph{parent} of such a vector is determined by its first entry
unequal to $\hash(\bot)$:
if that entry contains a direct certificate, the parent is
$(w-1,v^{\sf high},\pi^{\sf high})$; if it contains an indirect
certificate, the parent is
$(w_\star,v^{\sf high}_\star,\pi^{\sf high}_\star)$.
If no parent exists, set ${\sf Parent}(v):=(0,[\,],\bot)$.
We call the vector \emph{non-empty} if it induces a parent, and
\emph{empty} otherwise.

\subsubsection{Data Availability.}
Since \pcstr runs \pcver on vectors of hashes, we require Availability
(\fullversionref{def:pc:availability}) so that every decided hash is retrievable from
some honest party.
We assume collision resistance and a canonical encoding for hashed
proposal objects.
A proof for \Cref{alg:pc} and its \pcver variant is given in
\fullversionref{lem:availability}.
We assume a standard dispersal layer under which honest holders retain and
serve preimages, and every certified preimage becomes available to all
honest parties within the corresponding proposal-collection window.
Its communication is included in the stated asymptotic bounds.

\subsection{Protocol Description.}
\label{sec:pcstr-leaderless:protocol}
The pseudocode is given in \Cref{alg:pcstr-leaderless}. The protocol proceeds in views.
In every view $w$, parties run the \pcver instance $\Pi^{\sf VPC}_w$ on a length-$n$ hash vector and obtain
a verifiable low/high pair with proofs.

\paragraph{View~1 (deciding the low).}
Each party $i$ runs $\Pi^{\sf VPC}_1$ on its \pcstr input vector $v_i^{\sf in}$.
Upon receiving a verifiable low output $(v^{\sf low}_{i,1},\pi^{\sf low}_{i,1})$, it broadcasts $\newcommit=(1,v^{\sf low}_{i,1},\pi^{\sf low}_{i,1})$.
By the ${\sf Commit}$ procedure, any valid $\newcommit$ in view~1 makes
parties output $(v^{\sf low},\pi^{\sf low})
:=(v^{\sf low}_{i,1},\pi^{\sf low}_{i,1})$.

\paragraph{Views $w\ge 2$ (agreeing on the view-1 high).}
Views $w\ge 2$ carry no new inputs.
Instead, each party proposes a \emph{certificate} that is publicly verifiable and derived from earlier views.
Concretely, a party enters view $w$ upon receiving a valid $\newview=(w,\cert)$, rebroadcasts it once, starts a timeout,
and stores the sender's proposal object (\newview) in a buffer $P_{i,w}[\cdot]$.

To run $\Pi^{\sf VPC}_w$, party $i$ applies one cyclic shift to the
previous view's ranking, orders senders by $\rank^{\sf SPC}_{i,w}$, and
forms
\[
v^{\sf in}_{i,w}=[\,\hash(P_{i,w}[p_1]),\ldots,\hash(P_{i,w}[p_n])\,],
\]
using $\hash(\bot)$ for missing entries. Party $i$ invokes $\Pi^{\sf VPC}_w$ either after receiving all $n$ proposal objects
or when the view timer expires, whichever occurs first; each view instance
is invoked at most once.
Messages delivered by the timer deadline are processed before the timeout
event.
In a synchronous-minus-one execution, this deadline is the end of the
common logical proposal round: all correct parties freeze their buffers at
that round boundary before invoking \pcver.

Upon a low output $(v^{\sf low}_{i,w},\pi^{\sf low}_{i,w})$, party $i$ broadcasts $\newcommit=(w,v^{\sf low}_{i,w},\pi^{\sf low}_{i,w})$.
Upon receiving any valid $\newcommit$, a party calls
${\sf Commit}(w,v,\pi)$.
If $v$ has a parent, ${\sf Commit}$ follows parent pointers together with
their proofs until reaching view~1 and outputs the resulting
$(v^{\sf high},\pi^{\sf high})$ pair.
Each party outputs its low and high pairs at most once.

\paragraph{Direct advance vs.\ skipping.}
If $\Pi^{\sf VPC}_w$ outputs a high value $v^{\sf high}_{i,w}$ that \emph{has a parent} (i.e., ${\sf HasParent}(w,v^{\sf high}_{i,w})={\sf true}$),
party $i$ broadcasts a direct certificate for the next view.

Otherwise, party $i$ broadcasts an $\emptyview$ message containing its current highest non-empty high triple
$(w^{\sf high}_\star,v^{\sf high}_\star,\pi^{\sf high}_\star)$ and a
signature on
$(w,w^{\sf high}_\star,
\hash(v^{\sf high}_\star,\pi^{\sf high}_\star))$.
Once a party collects $f{+}1$ valid $\emptyview$ messages for the same view $w$, it aggregates the signatures into $\Sigma$,
selects any reported triple with maximum view index, forms an indirect certificate
$\cert^{\sf ind}(w,(w_\star,v_\star,\pi_\star),\Sigma)$, and broadcasts $\newview=(w{+}1,\cert)$.

\paragraph{Remark (Verifiability).}
\Cref{sec:accountability:pcstr} formally defines \pcverstr.
Our construction satisfies that definition: ${\sf Commit}$ preserves the
view-1 low proof and returns the selected view-1 high with its proof.

\paragraph{Safety intuition.}
Within each view, \pcver guarantees that verifiable lows are prefix-consistent and verifiable highs are publicly checkable extensions.
Across views, certificates are only accepted if they reference verifiable highs that themselves have parents (via ${\sf ValidCert}$),
so parent pointers always point to earlier verifiable highs and never to the sentinel $(0,[\,])$.
Moreover, the skip mechanism is conservative: an indirect certificate that jumps back to an earlier view can only occur when all intermediate
views are parentless (cf.\ \fullversionref{lem:pcstr:cert:1,lem:pcstr:cert:2}), so it cannot bypass any non-empty committed low.
As a result, all non-empty commits induce a single chain back to view~1, and as shown in \fullversionref{lem:pcstr:low-commit}, all honest parties output
the same view-1 high value.

\paragraph{Liveness intuition.}
After GST, honest parties keep entering new views (\fullversionref{lem:pcstr:enter}): each view either yields a non-empty high (direct advance)
or yields enough $\emptyview$ messages to form an indirect certificate (skip).
Cyclic shifting ensures that in some post-GST view $w^\star$ the first two
ranked parties are honest.
At most one is suspended, so honest inputs agree on a non-$\hash(\bot)$
proposal within the first two coordinates (and consistently use
$\hash(\bot)$ for a suspended proposer).
Validity therefore yields a low with a parent.
This triggers ${\sf Commit}$, which reaches view~1 and outputs the unique
view-1 high value.

\subsubsection{Protocol Analysis.}
Due to space constraints, we defer the proof that \Cref{alg:pcstr-leaderless} solves \pcstr (\Cref{def:pcstr}) and achieves Leaderless Termination to \fullversionref{sec:pcstr-leaderless:analysis}.

\section{Accountability}
\label{sec:accountability}

This section defines accountability for both \pcver and the verifiable
\pcstr construction of \Cref{sec:pcstr-leaderless}.
In the \smr application, accountability identifies the Byzantine proposer
responsible for a truncated output and allows parties to update the
proposer ranking.

\subsection{\pcacc}
\label{sec:accountability:pc}

Suppose that in \pc, the vector coordinates correspond to proposers ordered by a public ranking
$\rank=(p_1,\ldots,p_L)$.
Our \pc construction provides the following additional property: for any
verified truncated output $v$, the first excluded proposer
$\rank[|v|+1]$ is Byzantine.
Briefly, when honest parties agree on the value of every honest proposer,
that value reaches the threshold in every \qcone, so an output cannot
truncate immediately before an honest proposer.

\begin{definition}[Honest-Proposal Consistency]
\label{def:pc:honest-proposal-consistency}
An execution satisfies \emph{honest-proposal consistency} if, for every
honest proposer $p_k$, all honest input vectors contain the same value at
coordinate $k$.
\end{definition}

The \smr construction uses this condition for stabilized post-GST slots;
its proposal-collection argument is given in \Cref{sec:smr}.
For an instance $\id$, let $\rank_{\id}$ denote its proposer ranking.
Using canonical encoding, the identifier $\id$ commits to this ranking,
and every protocol message is signed together with $\id$.

\begin{definition}[\pcacc]\label{def:pcacc}
\pcacc satisfies \pcver and has the public deterministic blame function
\[
{\sf Blame}_{\id}(v,\pi):=
\begin{cases}
\rank_{\id}[|v|+1],
& \begin{array}{l}
  |v|<L\text{ and }
  \mathcal F_{\id}^{\sf low}(v,\pi)\vee
  \mathcal F_{\id}^{\sf high}(v,\pi),
  \end{array}\\[2mm]
\bot, & \text{otherwise.}
\end{cases}
\]
Thus, for a publicly accepted truncated output, ${\sf Blame}_{\id}$
returns its first excluded proposer.
It satisfies \emph{Accountability}: in every execution satisfying
honest-proposal consistency,
${\sf Blame}_{\id}(v,\pi)=p\neq\bot$ implies that $p$ is Byzantine.
\end{definition}

\begin{theorem}\label{thm:pc:accountable}
The \pcver construction of \Cref{thm:pc:verifiable} solves \pcacc
without additional messages, proofs, or communication.
\end{theorem}
\begin{proof}
By \Cref{thm:pc:verifiable}, the construction satisfies \pcver.
The function ${\sf Blame}_{\id}$ is public and deterministic because the
output proof and proposer ranking are both bound to $\id$.

Suppose ${\sf Blame}_{\id}(v,\pi)=p\neq\bot$ in an execution satisfying
honest-proposal consistency.
Then $|v|<L$ and $p=\rank_{\id}[k]$ for $k=|v|+1$.
Assume for contradiction that $p_k$ is honest.
Every honest \voteone sender has the same value $z$ at coordinate $k$,
and every valid \qcone contains at least
$t=\lfloor(n-f)/2\rfloor+1$ honest votes
(\Cref{lem:pc:qc1-threshold}).
Hence every $\x={\sf QC1Certify}(\qcone)$ that reaches coordinate $k-1$
also appends $z$ at coordinate $k$.

Here $\pi=\qcthree$.
By \Cref{lem:pc:1}, its $\xmcp$ values are pairwise prefix-consistent and
therefore totally ordered by $\preceq$.
Thus the low and high outputs are actual shortest and longest values in
\qcthree, and some valid \qctwo in \qcthree certifies $v$.
Every $\x$ in this \qctwo extends $v$ and therefore appends $z$ at
coordinate $k$.
Their maximum common prefix has length at least $k$, contradicting
$|v|=k-1$.

Thus the construction satisfies Accountability and solves \pcacc.
It reuses the existing \qcthree proof and therefore adds no messages,
certificate data, or communication.
\end{proof}

\subsection{Accountable \pcstr}
\label{sec:accountability:pcstr}

\begin{definition}[\pcverstr]\label{def:pcverstr}
\pcverstr satisfies \pcver and additionally ensures Agreement on the high
output:
\[
v_i^{\sf high}=v_j^{\sf high}
\qquad\text{for all honest }i,j\in\cH.
\]
\end{definition}

\begin{definition}[\pcaccstr]\label{def:pcaccstr}
\pcaccstr simultaneously satisfies \pcacc and \pcverstr for the same
input, verifiable outputs, and instance identifier.
\end{definition}

Equivalently, \pcaccstr is \pcacc with high-output Agreement, or
\pcverstr with Accountability.

\paragraph{Construction.}
No protocol changes to \Cref{alg:pcstr-leaderless} are required.
Instantiate its view-1 \pcver with \Cref{alg:pc}.
A communication-optimized implementation is given in
\fullversionref{sec:pc:optimized}.
The verifiable \pcstr low is already returned with its view-1 low proof.
Later views agree on a view-1 high, and ${\sf Commit}$ returns that high
with its original proof.
Thus both verifiable outputs belong to the same view-1 instance
\[
\id_1:=\hash(\sid,1,\rank^{\sf SPC}_1).
\]
For these outputs, define
\[
{\sf Blame}^{\sf SPC}_{\id_1}(v,\pi):={\sf Blame}_{\id_1}(v,\pi).
\]

\begin{theorem}\label{thm:pcstr:accountable}
Under partial synchrony, the construction of
\Cref{alg:pcstr-leaderless} solves \pcaccstr whenever its view-1 inputs
satisfy honest-proposal consistency.
It adds no messages or communication to the verifiable \pcstr
construction, and requires no additional proof data.
\end{theorem}
\begin{proof}
As shown in \fullversionref{thm:pcstr:correct}, the construction satisfies all \pcstr
properties, including Agreement on the high output.
Because both outputs retain proofs accepted by the view-1 predicates, the
construction also satisfies \pcverstr.
Under honest-proposal consistency, both the basic and optimized view-1
instantiations satisfy \pcacc
(\Cref{thm:pc:accountable} and
\fullversionref{thm:pc:optimized-correct}).
Since the verifiable \pcstr outputs retain their accepted view-1 proofs,
${\sf Blame}^{\sf SPC}_{\id_1}$ identifies a Byzantine proposer whenever
either output is truncated.
The proof data is already present in view-1 outputs and parent
certificates, so the construction adds no messages or communication.
\end{proof}

\section{\smr}
\label{sec:smr}

We construct \smr by running one \pcaccstr instance per slot,
sequentially.
The agreed high output finalizes the slot, while the low output allows
honest proposals to be committed earlier.
Truncated high outputs and reported low outputs provide accountable
evidence for updating the proposer ranking.

The protocol achieves $f$-Censorship Resistance and amortized commit
latency
\[
\tau_{\sf prop}+3\delta,
\]
where $\tau_{\sf prop}=2\Delta$ after GST and
$\tau_{\sf prop}=\Delta$ under synchronized slot starts.
In a failure-free execution, all proposals arrive in one actual delay, so
$\tau_{\sf prop}=\delta$.
Thus synchronized slots have amortized four-round commit latency.

\subsection{Definition}
\label{sec:smr:def}

Fix a public initial proposer ranking $\rank^{\sf MC}_1$.
All objects are implicitly bound to the \smr instance.
Slot number $s$ identifies the slot's \pcaccstr instance.
Its view-1 \pcacc identifier is
\[
\id_s:=\hash(s,1,\rank^{\sf MC}_s).
\]
We write $\Pi^{\sf ASPC}_s$ for this \pcaccstr instance.
For latency, the start of slot $s$ is the time by which all honest parties
have invoked ${\sf NewSlot}(s)$.

\paragraph{Proposal objects and evidence.}
Party $p$'s signed proposal object is
\[
P_{p,s}:=(s,p,b^{\sf in}_{p,s},e_{p,s-1}).
\]
A proposal object is valid if its signature verifies under $p$ and its
instance, slot, and proposer fields are correct; malformed evidence is
ignored without invalidating the proposal itself.
Proposal objects use canonical encoding and collision-resistant hashing.
For $P=(s,p,b,e)$, write $P.b:=b$ and $P.e:=e$.
The evidence field is either $\bot$ or
\[
e_{p,s-1}:=(v^{\sf low}_{p,s-1},\pi^{\sf low}_{p,s-1}),
\]
where the low is shorter than the agreed high of slot $s-1$.
Evidence is accepted only in the immediately following slot and only
under the source predicate for $\id_{s-1}$.

Party $i$ stores received proposal objects in $B_{i,s}[p]$, using $\bot$
when no valid object from $p$ has been received.
With $\rank^{\sf MC}_{i,s}=(p_1,\ldots,p_n)$, its slot input is
\[
v^{\sf in}_{i,s}
:=
[\,\hash(B_{i,s}[p_1]),\ldots,\hash(B_{i,s}[p_n])\,].
\]
The \pcaccstr instance returns verifiable outputs
\[
(v^{\sf low}_{i,s},\pi^{\sf low}_{i,s}),
\qquad
(v^{\sf high}_{i,s},\pi^{\sf high}_{i,s}).
\]

\paragraph{Ranking update.}
Each slot-$s$ proposal may carry a verifiable low from slot $s-1$.
After slot $s$ is finalized, let $\mathcal L_{s-1}$ contain the valid
previous-slot lows carried by proposal objects in the agreed slot-$s$
high; set $\mathcal L_0:=\emptyset$.
For $\rank=(p_1,\ldots,p_n)$ and a set $D$, define
${\sf Demote}(\rank,D)$ by moving the parties in $D$ to the end while
preserving everyone's relative order.

Let
\[
b_s:={\sf Blame}^{\sf SPC}_{\id_s}
       (v^{\sf high}_{i,s},\pi^{\sf high}_{i,s})
\]
be the proposer blamed by the current high, and let
\[
D_{s-1}:=
\{{\sf Blame}^{\sf SPC}_{\id_{s-1}}(v,\pi):
  (v,\pi)\in\mathcal L_{s-1}\}\setminus\{\bot,b_s\}.
\]
Thus $D_{s-1}$ contains the distinct proposers blamed by valid
previous-slot low evidence, excluding $b_s$.

The ranking update has two steps:
\[
\begin{aligned}
\rank'&:={\sf Demote}(\rank^{\sf MC}_{i,s},D_{s-1}),\\
\rank^{\sf MC}_{i,s+1}
&:={\sf Demote}(\rank',\{b_s\}\setminus\{\bot\}).
\end{aligned}
\]
The first step handles slow-low behavior from slot $s-1$.
Because honest parties may observe different lows, they carry them into
slot $s$ and apply only the common set selected by the agreed high.
Each valid truncated low blames a Byzantine proposer; demoting these
proposers yields the amortized commit-latency guarantee
(\fullversionref{thm:smr:complexity}).

The second step moves the proposer blamed by the current high to the
absolute end.
This keeps a current censor behind every honest proposer, even if stale
boundary evidence is present.
Consequently, each censored slot demotes a Byzantine censor behind all
honest proposers, yielding $f$-Censorship Resistance
(\fullversionref{thm:smr:2}).
Because the high output and its proposal objects are agreed, all honest
parties compute the same update.

\begin{algorithm*}[t!]
\caption{\smr Protocol for party $i$}
\label{alg:smr}
\centering
\vspace{-0.6em}
\begin{algorithmic}[1]

\begin{multicols}{2}
\raggedcolumns
\setlength{\columnsep}{1.2em}

\Statex {\bf Local variables (party $i$):}
\Statex $s_i$: current slot, initialized to $0$.
\Statex $B_{i,s}[\cdot]$: received slot-$s$ proposal objects, initialized to $\bot$.
\Statex $e_{i,s}$: evidence produced in slot $s$, initialized to $\bot$.

\Function{\sf Demote}{$\rank,D$}\label{line:smr:update}
    \State \Return $\rank$ with parties in $D$ moved to the end,
        preserving relative order
\EndFunction
\Function{\sf ValidEvidence}{$s,e$}
    \If{$s=1$ or $e=\bot$} \State \Return {\sf false} \EndIf
    \State unpack $e=(v,\pi)$
    \State \Return $\cF^{\sf low}_{\id_{s-1}}(v,\pi)
        \wedge v\prec v^{\sf high}_{i,s-1}$
\EndFunction
\Function{\sf CommittedLows}{$s,v$}
    \State $\mathcal L:=\emptyset$
    \For{$k=1,\dots,|v|$ with $v[k]\neq\hash(\bot)$}
        \State fetch the proposal object $P$ with $\hash(P)=v[k]$
        \If{${\sf ValidEvidence}(s,P.e)$}
            \State unpack $P.e=(v',\pi')$
            \State $\mathcal L:=\mathcal L\cup\{(v',\pi')\}$
        \EndIf
    \EndFor
    \State \Return $\mathcal L$
\EndFunction
\Function{\sf UpdateRank}{$s,v,\pi$}
    \State $b:={\sf Blame}^{\sf SPC}_{\id_s}(v,\pi)$
    \State $\mathcal L:={\sf CommittedLows}(s,v)$
    \State $D:=\{{\sf Blame}^{\sf SPC}_{\id_{s-1}}(v',\pi'):
        (v',\pi')\in\mathcal L\}\setminus\{\bot,b\}$
    \State $\rank':={\sf Demote}(\rank^{\sf MC}_{i,s},D)$
    \State \Return ${\sf Demote}(\rank',\{b\}\setminus\{\bot\})$
\EndFunction

\Upon{protocol start}
    \State $\rank^{\sf MC}_{i,1}:=\rank^{\sf MC}_1$
    \State ${\sf NewSlot}(1)$
\EndUpon
\Upon{${\sf NewSlot}(s)$}
    \State $s_i:=s$
    \State form signed
        $P_{i,s}:=(s,i,b^{\sf in}_{i,s},e_{i,s-1})$
    \State $B_{i,s}[i]:=P_{i,s}$
    \State broadcast $\proposal:=(s,P_{i,s})$
    \State start $\timer(s)$ for $\tau_{\sf prop}$
\EndUpon
\Upon{receiving the first valid $\proposal=(s,P_{p,s})$ from $p$
       before slot $s$ is frozen}
    \State $B_{i,s}[p]:=P_{p,s}$
    \If{$B_{i,s}[j]\neq\bot$ for all $j\in[n]$}
        \State ${\sf RunASPC}(s)$
    \EndIf
\EndUpon
\Upon{$\timer(s)$ expires and $s=s_i$}
    \State ${\sf RunASPC}(s)$
\EndUpon
\Upon{first invocation of ${\sf RunASPC}(s)$}
    \State freeze $B_{i,s}$ and let
        $\rank^{\sf MC}_{i,s}=(p_1,\dots,p_n)$
    \State $v^{\sf in}_{i,s}:=
        [\,\hash(B_{i,s}[p_1]),\dots,\hash(B_{i,s}[p_n])\,]$
    \State run $\Pi^{\sf ASPC}_{s}$ with input $v^{\sf in}_{i,s}$
        and ranking $\rank^{\sf MC}_{i,s}$
\EndUpon
\Upon{$\Pi^{\sf ASPC}_{s}$ outputs
       $(v^{\sf low}_{i,s},\pi^{\sf low}_{i,s})$}
    \State ${\sf Commit}(s,\rank^{\sf MC}_{i,s},v^{\sf low}_{i,s})$
\EndUpon
\Upon{$\Pi^{\sf ASPC}_{s}$ outputs
       $(v^{\sf high}_{i,s},\pi^{\sf high}_{i,s})$}
    \State wait until the local low output is available
    \State ${\sf Commit}(s,\rank^{\sf MC}_{i,s},v^{\sf high}_{i,s})$
    \State $\rank^{\sf MC}_{i,s+1}:=
        {\sf UpdateRank}(s,v^{\sf high}_{i,s},\pi^{\sf high}_{i,s})$
    \If{$v^{\sf low}_{i,s}\prec v^{\sf high}_{i,s}$}
        \State $e_{i,s}:=(v^{\sf low}_{i,s},\pi^{\sf low}_{i,s})$
    \EndIf
    \State ${\sf NewSlot}(s+1)$
\EndUpon
\Upon{{\sf Commit}$(s,\rank,v)$}
    \For{$k=1,\dots,|v|$ with $v[k]\neq\hash(\bot)$}
        \State fetch and verify $P_{p,s}$ for $p=\rank[k]$ with
            $\hash(P_{p,s})=v[k]$
        \State commit $P_{p,s}.b$ at $(s,k)$ if not committed
    \EndFor
\EndUpon

\end{multicols}
\end{algorithmic}
\vspace{-0.6em}
\end{algorithm*}

\subsection{Protocol Description}
\label{sec:smr:protocol}

\paragraph{Starting a slot.}
At slot start, party $i$ broadcasts its signed proposal object carrying
the previous slot's evidence.
It accepts only the first valid object from each proposer.
The proposal buffer is frozen when \pcaccstr is first invoked.
Messages delivered at the proposal deadline are processed before the
timeout event.
After GST, previous-slot high agreement bounds honest slot-start skew by
$\Delta$, so a $2\Delta$ collection window gives honest-proposal
consistency.
For synchronized slot starts, a $\Delta$ window suffices.
We absorb the first post-GST evidence-cleanup slot into the stabilization
period.

\paragraph{Committing and advancing.}
Upon a low output, parties verify and early-commit its proposal objects.
Upon the agreed high output, they finalize the slot, process accountable
evidence contained in its committed proposal objects, compute the next
ranking, and start slot $s+1$.
If the local low is shorter than the agreed high, its proof is attached to
the next proposal object.

Procedure ${\sf Commit}(s,\rank_s,v)$ maps coordinate $k$ to proposer
$\rank_s[k]$, fetches and verifies the canonical proposal object, and
deduplicates by slot-position.
Low commits are safe prefixes of the final high output.

\paragraph{Safety and liveness intuition.}
Within a slot, accepted lows prefix the agreed high, so early commits are
consistent with finalization.
Slots are finalized sequentially, yielding one common log.
After GST, \pcaccstr Termination advances every slot.
Every honest proposer places its own signed proposal object in its local
input.
Once honest inputs are consistent, Availability prevents an accepted
output from replacing that proposal with another value.
Thus a censored slot has a truncated high
(\fullversionref{lem:smr:proposal-integrity}), whose proof identifies a Byzantine
proposer.
Similarly, a low shorter than the high identifies a Byzantine proposer
once its evidence is committed in the following slot.

\subsubsection{Protocol Analysis.}
Proofs are deferred to \fullversionref{sec:smr:analysis}.
Under synchrony and an honest first proposer in ASPC's view-2 ranking, the
protocol has good-case message complexity $O(n^2)$ per slot; its
worst-case message complexity is $O(n^3)$.

\section{Related Work}
\label{sec:related}

BFT consensus has been extensively studied, with foundational results establishing tight bounds on resilience and complexity under various network assumptions~\cite{dolev1982byzantine, dwork1988consensus}. Practical systems like PBFT~\cite{castro1999practical} introduced leader-driven architectures, inspiring optimizations in performance, modularity, and communication efficiency~\cite{kotla2007zyzzyva, gueta2019sbft, bessani2014bftsmart, cachin2016architecture}. 
Asynchronous protocols~\cite{miller2016honeybadger,abraham2019asymptotically,lu2020dumbo} achieve progress without timing assumptions.
Traditional partially synchronous BFT protocols often rely on designated leaders, making them vulnerable to targeted DoS attacks and providing weak inclusion guarantees~\cite{castro1999practical}. Leader rotation~\cite{cohen2022aware} 
mitigates some issues but does not prevent selective exclusion of honest inputs or targeted DoS vulnerability~\cite{gelashvili2022jolteon}.

\paragraph{Leaderless protocols} Antoniadis et al.~\cite{antoniadis2021leaderless} formalized leaderless consensus, where progress holds despite an adversary suspending one party per round. Protocols like DBFT~\cite{crain2018dbft} and asynchronous designs~\cite{miller2016honeybadger, guo2020dumbo, danezis2022narwhal, abraham2019asymptotically} are sometimes called leaderless but lack formal censorship-resistance as defined here and in~\cite{garimidi2025multiple, abraham2025latency}, allowing adversaries to censor honest proposals.

\paragraph{Multi-leader protocols} To mitigate the leader performance bottleneck and enhance throughput, multi-leader approaches~\cite{stathakopoulou2019mirbft,orthrus2024,ladon2024,iss2022} run parallel instances of single-leader BFT consensus such as PBFT~\cite{castro2002practical} and merge outputs.
However, these protocols do not provide formal censorship-resistance guarantees.
Moreover, they are not leaderless as per the definition of~\cite{antoniadis2021leaderless}.
Instead of avoiding vulnerability to faulty leaders, these protocols typically exacerbate it, as demonstrated experimentally in~\cite{hydra2025}.

\paragraph{Leader-based protocols with multiple proposers} 
Some works propose to decouple dissemination from agreement via a parallel dissemination layer, while relying on designated leaders for ordering~\cite{cohen2023proof, quorumstore}.
Refinements such as Autobahn~\cite{giridharan2024autobahn} and Raptr~\cite{tonkikh2025raptr} optimize latency and robustness. Raptr uses a prefix-based mechanism on block lengths, inspiring our Prefix Consensus definition, but does not formalize it or explore its applications to leaderless protocols nor censorship-resistance.


\paragraph{Censorship resistance.}
Two recent works~\cite{abraham2025latency,garimidi2025multiple} explicitly
study censorship in BFT consensus. Abraham et al.~\cite{abraham2025latency}
prove matching five-round lower and upper bounds for strong single-shot
censorship resistance. The upper bound is mainly theoretical. Their lower
bound applies to the single-slot analogue of our setting, where a transaction
submitted to one honest validator must be included in that slot. We instead
consider long-lived multi-slot consensus: by allowing only a bounded number
of post-GST censored slots ($f$ in our protocol), we rule out infinitely many
inclusion failures while avoiding the single-shot lower bound, achieving lower
amortized commit latency, and providing leaderless
termination~\cite{antoniadis2021leaderless}.
Garimidi et al. (or MCP)~\cite{garimidi2025multiple} use an all-or-nothing notion:
each slot outputs either all timely honest inputs or the empty set. This
prevents selective censorship in non-empty slots, but permits censorship via
skipping the slots entirely. In contrast, our definition counts any slot that omits an
honest input, including an empty slot, as censored. This distinction matters
in practice: skipped blocks leave transactions pending and, unless encrypted,
observable in the mempool, enabling a Byzantine next leader to exploit them for MEV. Moreover, MCP is not leaderless and remains vulnerable to
leader-targeted attacks.
MCP~\cite{garimidi2025multiple} additionally hides transaction contents until
ordering, but only with lower resilience ($f<n/5$) and additional trust
assumptions on relay nodes. Our protocol is orthogonal to privacy and could
obtain a similar hiding guarantee without weakening resilience or relying
on relay nodes, e.g., by integrating an encrypted mempool such
as~\cite{fernando2025trx}.

\paragraph{DAG-based protocols} Similarly to other works listed here, DAG BFT protocols~\cite{baird2016hashgraph,danielsson2018aleph,keidar2021dagrider,keidar2023cordial,spiegelman2022bullshark,spiegelman2023shoal,arun2024shoal++,sailfish2025, sailfishpp2025,babel2023mysticeti} distribute the workload of payload dissemination among all parties to achieve higher throughput and avoid the leader performance bottleneck. 
Moreover, their round-based structure ensures that at least $2f+1$ proposals are incorporated per round, providing a limited form of censorship resistance. 
However, they typically still rely on leaders (usually called \emph{anchors}) for inclusion (choosing the $2f+1$ proposals) and ordering. Thus, they inherit similar limitations with respect to targeted DoS attacks and censorship as leader-based and multi-leader BFT protocols.

\paragraph{Graded Consensus}
As discussed in \fullversionref{sec:discussion:graded}, Prefix Consensus is closely related to (multi-valued) Graded Consensus~\cite{connected-cons-2024}.
Graded consensus and similar primitives~\cite{connected-cons-2024,connected-cons-2026,abraham2024round,yang1998structured} are interesting in their own right, as they circumvent the famous FLP impossibility theorem~\cite{fischer1985impossibility} by relaxing the agreement property of consensus.
Additionally, they are often used as key building blocks in randomized asynchronous Byzantine Agreement protocols (e.g.,~\cite{abraham2022bindingcrusader_podc,bouzid2015minimal,deligios2021round,momose2022constant,blum2020asynchronous}).

\section{Conclusion}
\label{sec:conclusion}
This paper studies censorship resistance through the lens of a new primitive,
\emph{Prefix Consensus}. We show that \pc admits a deterministic,
purely asynchronous solution under optimal resilience, and we establish a tight
round complexity: three communication rounds are necessary and sufficient in the
asynchronous Byzantine setting with optimal resilience ($n = 3f+1$). Building on a verifiable variant, we
then construct a partially synchronous \smr protocol that is
\emph{leaderless} and achieves \emph{$f$-censorship resistance} after GST: at most
$f$ slots may omit some honest proposals, after which all honest proposals are
included in every slot.





%
%
%
\ifcameraready
  \bibliographystyle{ACM-Reference-Format}
\else
  \bibliographystyle{splncs04}
\fi
\balance
\bibliography{references}

\ifcameraready\else
  \clearpage
  \nobalance
  \appendix

  \section{Analysis of \pcstr Protocol}
\label{sec:pcstr-leaderless:analysis}

In this section, we analyze the properties of our \pcstr protocol (\Cref{alg:pcstr-leaderless}) presented in~\Cref{sec:pcstr-leaderless}.

\paragraph{Instantiation scope.}
The Agreement, Prefix, Validity, and ordinary Termination proofs
treat \pcver as a black box.
The Leaderless Termination proof instantiates \pcver with the
threshold-prefix protocol of \Cref{alg:pc}, while the complexity analysis
uses its communication-optimized implementation from
\Cref{sec:pc:optimized}.

\paragraph{Data availability.}\label{sec:data-availability}
Conceptually, the \pcstr construction invokes \pcver on vectors of full
proposal objects.
In this form, it is a black-box construction whose correctness uses only
the \pcver properties.
Algorithm~\ref{alg:pcstr-leaderless} replaces proposal objects by their
hashes solely to reduce communication.
For this optimization, every certified hash must be retrievable from an
honest holder; this motivates the Availability property below.
Every subsequent use of Availability can be omitted in the full-object
version, where the proposal and its certificate are read directly from the
\pcver output.

\begin{definition}[Availability]\label{def:pc:availability}
A \pc (resp., \pcver) protocol satisfies \emph{Availability} if for any honest (resp., publicly verifiable) output $v^{\sf low}$ and any index $1\le k\le |v^{\sf low}|$, there exists an honest party $h$ such that $v_h^{\sf in}[k]=v^{\sf low}[k]$. The same guarantee holds for any honest (resp., publicly verifiable) output $v^{\sf high}$.
\end{definition}

\begin{lemma}\label{lem:availability}
    \Cref{alg:pc} and its \pcver variant (\Cref{sec:pcver}) satisfy Availability.
\end{lemma}

\begin{proof}
    Consider any honest (or publicly verifiable) output produced from a valid $\qcthree$.
    Every coordinate of $v^{\sf low}$ (resp., $v^{\sf high}$) appears in
    some $\xmcp$ value certified in $\qcthree$, because these values are
    mutually consistent (\Cref{lem:pc:1}) and the outputs are their
    shortest and longest elements.
    Every coordinate of such an $\xmcp$ appears in every $\x$ value of
    its certifying $\qctwo$.
    Finally, by ${\sf QC1Certify}$, every coordinate appended to an $\x$
    value has support from
    $t=\lfloor(n-f)/2\rfloor+1>f$ distinct \voteone senders in its
    certifying \qcone.
    At least one supporter is honest and has the same input element at
    that coordinate.
    Thus every element of either output is present in an honest party's
    input, proving Availability.
\end{proof}


\paragraph{Useful lemmas.}

\begin{lemma}\label{lem:pcstr:cert:1}
If a valid indirect certificate for entering view $w$ is formed, then
every publicly accepted low value $v^{\sf low}_{w-1}$ in view $w-1$ has
no parent:
\[
{\sf Parent}(v^{\sf low}_{w-1})=(0,[\,],\bot).
\]
\end{lemma}

\begin{proof}
A valid indirect certificate for entering view $w$ requires $f+1$ valid \emptyview messages for view $w-1$
(\Cref{alg:pcstr-leaderless}), so at least one such message is from an honest party $h$.
Honest $h$ broadcasts \emptyview in view $w-1$ only upon receiving a verifiable high output
$(v^{\sf high}_{h,w-1},\pi^{\sf high}_{h,w-1})$ with ${\sf HasParent}(w-1,v^{\sf high}_{h,w-1})={\sf false}$.
The certificate bounds imply $w-1>1$, so
${\sf Parent}(v^{\sf high}_{h,w-1})=(0,[\,],\bot)$.

Let $v^{\sf low}_{w-1}$ be any publicly accepted low in the same view.
By \pcver Proof Soundness,
$v^{\sf low}_{w-1}\preceq v^{\sf high}_{h,w-1}$.
If $v^{\sf low}_{w-1}$ had a parent, its first non-$\hash(\bot)$ entry
would also occur in $v^{\sf high}_{h,w-1}$, giving the latter a parent---a
contradiction.
Therefore ${\sf Parent}(v^{\sf low}_{w-1})=(0,[\,],\bot)$.
\end{proof}

\begin{lemma}\label{lem:pcstr:cert:2}
Suppose a valid indirect certificate for entering view $w$ references a
parent high from view $w'<w-1$.
Then, for every $u$ with $w'<u<w$, every publicly accepted low
$v^{\sf low}_u$ in view $u$ has no parent:
\[
{\sf Parent}(v^{\sf low}_u)=(0,[\,],\bot).
\]
\end{lemma}

\begin{proof}
We use strong induction on $w$.
The base case $w=3$ follows from \Cref{lem:pcstr:cert:1}, since the only
possible referenced parent view is $w'=1$.

For the inductive step, let
\[
\cert^{\sf ind}(w-1,(w',v',\pi'),\Sigma)
\]
be a valid indirect certificate entering view $w$.
By \Cref{lem:pcstr:cert:1}, every publicly accepted low in view $w-1$ is
parentless.
It remains to consider views strictly between $w'$ and $w-1$.

Choose an honest signer $h$ among the $f+1$ reports in $\Sigma$, and let
$r_h\le w'$ be the latest parent view reported by $h$.
Before executing view $w-1$ and producing this report, $h$ entered
view $w-1$ using some valid certificate $\cert_h$.
Let $r$ be the parent view referenced by $\cert_h$.
When $h$ processed $\cert_h$, it updated its latest parent view to at
least $r$; hence
\[
r\le r_h\le w'.
\]

If $\cert_h$ is direct, then $r=w-2$.
The above inequality and $w'<w-1$ imply $w'=w-2$, so there are no views
strictly between $w'$ and $w-1$.

Otherwise, $\cert_h$ is an indirect certificate entering the smaller view
$w-1$ and referencing parent view $r$.
By the induction hypothesis, every publicly accepted low in each view
$u$ with $r<u<w-1$ is parentless.
Since $r\le w'$, this includes every $u$ with $w'<u<w-1$.
Combining this with the parentlessness of view $w-1$ proves the lemma.
\end{proof}

\begin{lemma}[Valid parent]\label{lem:pcstr:valid-parent}
Let $v$ be a publicly accepted output in view $w>1$.
If procedure ${\sf Parent}(v)$ returns
$(u,h,\pi)\neq(0,[\,],\bot)$, then:
\begin{enumerate}
    \item $0<u<w$;
    \item $\cF^{\sf high}_u(h,\pi)$ and ${\sf HasParent}(u,h)$ hold; and
    \item the first non-$\hash(\bot)$ entry of $v$ is the hash of a valid
    view-$w$ proposal object whose certificate references $(u,h,\pi)$.
\end{enumerate}
\end{lemma}
\begin{proof}
Let $k$ be the first coordinate of $v$ unequal to $\hash(\bot)$.
By Availability (\Cref{lem:availability}), there is an honest party whose
view-$w$ input has value $v[k]$ at coordinate $k$.
By the input construction, this value is
$\hash(P_k)$ for a proposal object $P_k=(w,\cert)$ accepted by
${\sf ValidCert}(w,\cert)$.
Collision resistance and canonical encoding ensure that ${\sf Parent}$
fetches this same object, so $\cert$ is the certificate from which
${\sf Parent}(v)$ obtains $(u,h,\pi)$.

If $\cert$ is direct, it has form
$\cert^{\sf dir}(w-1,h,\pi)$ and hence $u=w-1$.
If it is indirect, it has form
$\cert^{\sf ind}(w-1,(u,h,\pi),\Sigma)$ and validation gives
$0<u<w-1$.
In both cases, ${\sf ValidCert}$ also checks
$\cF^{\sf high}_u(h,\pi)$ and ${\sf HasParent}(u,h)$.
These facts establish all three claims.
\end{proof}

\paragraph{Roots.}
\begin{definition}[Root]\label{def:pcstr:root}
For every publicly accepted view-1 high $h$, define
${\sf Root}(1,h):=h$.
For a publicly accepted parented value $v$ in view $w>1$, define
\[
{\sf Root}(w,v):={\sf Root}(u,h)
\qquad\text{where }{\sf Parent}(v)=(u,h,\pi).
\]
\end{definition}

\begin{lemma}[Root is well-defined]\label{lem:pcstr:root-defined}
The root of every publicly accepted parented value is a publicly accepted
view-1 high value.
\end{lemma}
\begin{proof}
By \Cref{lem:pcstr:valid-parent}, every parent step moves to a strictly
smaller positive view and carries a publicly accepted high proof.
If the parent view is greater than one, ${\sf ValidCert}$ requires that
high to have another parent.
The recursion therefore reaches view~1 and cannot stop earlier.
\end{proof}

\begin{lemma}[Same-view parent]\label{lem:pcstr:same-parent}
Any two publicly accepted non-empty lows in the same view have the same
parent.
Moreover, every publicly accepted high in that view has the same parent
as any publicly accepted non-empty low.
\end{lemma}
\begin{proof}
Let $L,L'$ be two publicly accepted non-empty lows in the same view.
The \qcthree proof certifying $L$ also certifies a publicly accepted high
$H$ in that view.
By \pcver Proof Soundness, $L\preceq H$ and $L'\preceq H$.
Since $L$ is non-empty, $H$ is non-empty as well.
Both lows contain the first non-$\hash(\bot)$ entry of $H$, so that entry
is identical in $L$ and $L'$ and their parents are equal.
The same first-entry argument gives
${\sf Parent}(L)={\sf Parent}(H)$.
\end{proof}

\begin{lemma}[Certificate lock]\label{lem:pcstr:cert-extend}
Let $L$ be a publicly accepted non-empty low in view $w\ge2$, and let
$R={\sf Root}(w,L)$.
For every $x>w$, every valid certificate entering view $x$ references a
high $H$ from a view $u\ge w$ such that ${\sf Root}(u,H)=R$.
\end{lemma}
\begin{proof}
We use strong induction on $x$.

For $x=w+1$, an indirect certificate would, by
\Cref{lem:pcstr:cert:2}, imply that $L$ is parentless.
Thus the certificate is direct and references a publicly accepted high
$H$ in view $w$.
By \Cref{lem:pcstr:same-parent},
${\sf Parent}(H)={\sf Parent}(L)$, so ${\sf Root}(w,H)=R$.

Now let $x>w+1$ and assume the claim for every smaller target view.
If the certificate is direct, it references a parented high $H$ in view
$x-1$.
By \Cref{lem:pcstr:valid-parent}, the first non-$\hash(\bot)$ entry of
$H$ contains a valid proposal object with a certificate entering view
$x-1$.
The induction hypothesis gives root $R$ to the high referenced by that
certificate.
Writing ${\sf Parent}(H)=(r,G,\pi_G)$, that referenced high is $G$, so
${\sf Root}(x-1,H)={\sf Root}(r,G)=R$.

If the certificate is indirect, let $(u,H,\pi_H)$ be its referenced high.
If $u<w$, \Cref{lem:pcstr:cert:2} would again imply that $L$ is
parentless, so $u\ge w$.
For $u=w$, \Cref{lem:pcstr:same-parent} gives root $R$.
For $u>w$, \Cref{lem:pcstr:valid-parent} exposes a valid certificate
entering view $u$ in the proposal object at the first
non-$\hash(\bot)$ entry of $H$.
The induction hypothesis gives root $R$ to the high referenced by that
certificate.
Writing ${\sf Parent}(H)=(r,G,\pi_G)$, that referenced high is $G$, so
${\sf Root}(u,H)={\sf Root}(r,G)=R$.
\end{proof}

\begin{lemma}\label{lem:pcstr:low-extend}
Any two publicly accepted non-empty lows in views $2\le w\le w'$ have the
same root.
\end{lemma}
\begin{proof}
If $w=w'$, the claim follows from \Cref{lem:pcstr:same-parent}.
If $w<w'$, let $L$ and $L'$ be the two lows.
The proposal object at the first non-$\hash(\bot)$ entry of $L'$ contains
a valid certificate entering view $w'$.
By \Cref{lem:pcstr:cert-extend}, the high referenced by that certificate
has root ${\sf Root}(w,L)$.
This high is the parent of $L'$, so
${\sf Root}(w',L')={\sf Root}(w,L)$.
\end{proof}

\begin{lemma}\label{lem:pcstr:low-commit}
For any two publicly accepted non-empty lows $L,L'$ in views
$w,w'\ge2$, ${\sf Commit}(w,L,\cdot)$ and
${\sf Commit}(w',L',\cdot)$ output the same view-1 high value.
\end{lemma}
\begin{proof}
Procedure ${\sf Commit}$ follows exactly the parent recursion defining
${\sf Root}$ and therefore outputs the root of its input low.
By \Cref{lem:pcstr:low-extend}, the two lows have the same root, so the two
calls output the same view-1 high.
\end{proof}

\begin{lemma}\label{lem:pcstr:enter}
After GST, honest parties keep entering new views.
\end{lemma}
\begin{proof}
For contradiction, suppose there is a last view $w_{\sf last}$ such that some honest party enters $w_{\sf last}$,
but no honest party ever enters any view $>w_{\sf last}$.
Once an honest party enters view $w_{\sf last}$, it relays a valid \newview message, so all honest parties enter $w_{\sf last}$
within bounded time after GST.
If any party advances before every view timer expires, then an honest
party enters a higher view, contradicting the choice of $w_{\sf last}$.
Otherwise, all timers eventually expire and every honest party starts
$\Pi^{\sf VPC}_{w_{\sf last}}$.

If any honest party obtains a verifiable high output $v^{\sf high}$ with ${\sf HasParent}(w_{\sf last},v^{\sf high})={\sf true}$,
it broadcasts a direct \newview for $w_{\sf last}+1$, contradiction.
Otherwise, every honest party broadcasts an \emptyview message for $w_{\sf last}$, and some honest party collects $f+1$ such messages,
forms an indirect certificate for $w_{\sf last}+1$, and broadcasts \newview, again contradiction.
\end{proof}

\begin{theorem}\label{thm:pcstr:correct}
\pcstr Protocol (\Cref{alg:pcstr-leaderless}) solves \pcstr (\Cref{def:pcstr}) under partial synchrony:
Agreement, Prefix, and Validity always hold, and Termination holds
after GST.
\end{theorem}

\begin{proof}
\emph{Validity.}
Each honest party outputs a publicly accepted view-1 low.
The same \qcthree proof also certifies a publicly accepted view-1 high.
By \pcver Proof Soundness,
\[
\mcp{\{v_h^{\sf in}\}_{h\in\cH}}\preceq v_i^{\sf low}.
\]

\emph{Prefix.}
Fix honest parties $i,j$.
The verifiable outputs
$(v_i^{\sf low},\pi_i^{\sf low})$ and
$(v_j^{\sf high},\pi_j^{\sf high})$ are both publicly accepted by the
view-1 \pcver predicates.
Proof Soundness therefore gives
$v_i^{\sf low}\preceq v_j^{\sf high}$.

\emph{Agreement.}
Consider any two honest parties $i,j$ that output $v_i^{\sf high}$ and $v_j^{\sf high}$.
Each high is produced by committing a publicly accepted non-empty low
from some view at least two and following its parent chain.
By \Cref{lem:pcstr:low-commit}, any two such calls output the same view-$1$ high value.
Hence $v_i^{\sf high}=v_j^{\sf high}$.

\emph{Termination.}
By \pcver Termination, every honest party eventually outputs its view-1
low pair.
It remains to show that every honest party eventually outputs $v_i^{\sf high}$.

By \Cref{lem:pcstr:enter}, after GST honest parties keep entering new views.
Within at most $f+1$ cyclic shifts, an honest party is first in the
ranking.
Fix such a post-GST view $w^\star$.
Its proposal object reaches every honest party before the proposal
deadline, so every honest input to $\Pi^{\sf VPC}_{w^\star}$ has the same
non-$\hash(\bot)$ first coordinate.
The honest-input maximum common prefix therefore has a parent.
By \pcver Validity, every publicly accepted low in this view has a parent.
The parties broadcast these lows, invoke ${\sf Commit}$, and by
\Cref{lem:pcstr:root-defined} reach a view-1 high.
\end{proof}

\begin{theorem}\label{thm:pcstr:leaderless}
\pcstr Protocol (\Cref{alg:pcstr-leaderless}) achieves Leaderless Termination (\Cref{def:leaderless})
against an adversary that can suspend up to one party per round while up
to $f-1$ other parties are Byzantine.
As in the synchronous-minus-one message-passing model, each proposal
phase is one synchronized logical round: a non-suspended honest proposer
reaches every correct party, while a suspended proposer reaches none.
Suspension pauses a correct party for that round; messages remain pending
on the reliable channels and are delivered when it next takes a step.
\end{theorem}

\begin{proof}
There are at least $n-f+1$ honest parties and at most $f-1$ Byzantine
parties, so the cyclic ranking contains two adjacent honest parties.
In every \pcver phase, at least $n-f$ honest parties are active; pending
messages let every correct party catch up when it next takes a step.
Thus the view-1 instance terminates and every correct party outputs its
low value, and the view-entry argument of \Cref{lem:pcstr:enter} continues
to apply.

As the ranking shifts, some view $w^\star$ places such a pair in the first
two positions.

During the proposal round of view $w^\star$, at most one party is
suspended.
Let $A$ be a set of $n-f$ non-suspended honest parties in that round.
If the first proposer is active, every input in $A$ contains its same
non-$\hash(\bot)$ proposal in coordinate one.
If the first proposer is suspended, every input in $A$ contains
$\hash(\bot)$ in coordinate one, while the second proposer is active and
has the same non-$\hash(\bot)$ proposal in coordinate two.
Thus the inputs in $A$ share a prefix containing a parent.

The mobile suspension of one party in each \pcver round leaves at least
$n-f$ active honest parties, enough to form every required quorum and
advance the \pcver instance.
Moreover, every \qcone of size $n-f$ intersects $A$ in at least
$n-2f\ge t$ parties.
The threshold-prefix rule therefore forces every certified $\x$, and
hence every accepted low in view $w^\star$, to contain the common parent
prefix.
Every correct party eventually receives a committed low from this view;
by \Cref{lem:pcstr:root-defined}, ${\sf Commit}$ follows its parent chain
to a view-1 high.
Hence every correct party eventually outputs.
\end{proof}

Now we analyze the performance of~\Cref{alg:pcstr-leaderless}. 

\begin{theorem}\label{thm:pcstr:complexity}
\pcstr Protocol (\Cref{alg:pcstr-leaderless}), under a synchronized
post-GST start, has latency $3\delta$ for $v_i^{\sf low}$ and
$v_i^{\sf high}$ within
\[
2(f+1)\Delta+(5f+7)\delta
\]
in the worst case.
Under synchrony and no failures, it outputs $v_i^{\sf high}$ within $7$ rounds ($7\delta$).

Under synchrony, it has worst-case message complexity $O(n^3)$ and
communication complexity $O((\hashsize+\sigsize)n^4)$,
where $\hashsize$ is the hash size and $\sigsize$ is the signature size.
Under synchrony, when the first party in the view-2 ranking is honest,
it has message complexity $O(n^2)$ and communication complexity $O((\hashsize+\sigsize)n^3)$.
\end{theorem}

\begin{proof}
\emph{Latency.}
The low value is the verifiable low output of $\Pi^{\sf VPC}_1$, which is obtained in $3$ rounds, hence $3\delta$.
Delivering the initial direct \newview takes at most $\delta$.
Let $k$ be the number of attempted views after view~1.
Each such view takes at most $2\Delta+3\delta$: $2\Delta$ for proposal
collection and $3\delta$ for \pcver.
Between two attempts, collecting \emptyview messages and delivering the
resulting indirect \newview take at most $2\delta$.
Consequently,
\[
T_{\sf high}
\le
3\delta+\delta+k(2\Delta+3\delta)+2(k-1)\delta.
\]

A parentless low can occur only while a Byzantine party is first in the
cyclic ranking.
Each failed view shifts that party away from the first position, so
$k\le f+1$.
Therefore
\[
T_{\sf high}
\le 2(f+1)\Delta+(5f+7)\delta.
\]
In the fault-free synchronous case, parties exchange view-$2$ proposal objects in one round and then run $\Pi^{\sf VPC}_2$,
which takes $3$ rounds. Thus within $4$ additional rounds, some non-empty value is committed in view~2 and ${\sf Commit}$
outputs the (view-$1$) high value. The total high latency is $3+4=7$ rounds, i.e., $7\delta$.

\emph{Message and communication complexity.}
Per view, each party broadcasts $O(1)$ messages of each outer-protocol
type, for a total of $O(n^2)$ messages.
Message payloads consist of hash vectors, \pcver proofs, and (possibly aggregate) signatures. When $L=n$, a \pcver proof has size
$O(\hashsize\cdot n+\sigsize\cdot n+n)$ (by the optimization of~\Cref{sec:pc:optimized}), so certificates and proposal objects have the same asymptotic size.
Hence auxiliary communication per view is
$O(n^2(\hashsize\cdot n+\sigsize\cdot n+n))
=O((\hashsize+\sigsize)n^3)$.

Each \pcver invocation has message complexity $O(n^2)$ and communication complexity $O((\hashsize+\sigsize)n^3)$ for $L=n$
(\Cref{thm:pc:optimized:complexity}). Adding the auxiliary communication preserves the asymptotics, yielding overall message complexity
$O(n^2)$ and communication complexity $O((\hashsize+\sigsize)n^3)$ per view.

Since in the worst case, there can be $O(n)$ views, we have the message complexity $O(n^3)$ and communication complexity $O((\hashsize+\sigsize)n^4)$.
Under synchrony, when the first party in the view-2 ranking is honest, the
protocol finishes in 2 views, therefore
it has message complexity $O(n^2)$ and communication complexity $O((\hashsize+\sigsize)n^3)$.
\end{proof}

  \section{Analysis of \smr Protocol}
\label{sec:smr:analysis}

\begin{lemma}[Post-GST proposal consistency]
\label{lem:smr:proposal-consistency}
For every sufficiently post-GST slot, honest slot-start times differ by at
most $\Delta$, and the $2\Delta$ proposal window establishes
honest-proposal consistency.
\end{lemma}
\begin{proof}
An honest party starts slot $s$ only after obtaining the agreed high of
slot $s-1$.
The message that first causes an honest party to obtain this high is
relayed to every honest party and delivered within $\Delta$.
By the \pcstr protocol, view-1 lows are broadcast and relayed, and
\Cref{alg:smr} waits for the local low before processing the high.
The dispersal assumption makes all certified preimages available, so
local commit, evidence processing, and ranking update add no network delay.
Thus every honest party starts slot $s$ within $\Delta$ of the first.

Every honest proposer broadcasts immediately upon starting the slot.
Its proposal therefore reaches the earliest-starting honest party within
another $\Delta$ and is processed before the $2\Delta$ timeout.
All later-starting parties receive it no later than their own deadline.
Each honest proposer also inserts its own proposal locally, proving
honest-proposal consistency.
\end{proof}

We absorb the first consistent slot and its one-slot evidence cleanup into
GST, under the effective-GST convention of \Cref{sec:def:notation}.
Thus every slot considered below, and every low used as evidence in such a
slot, satisfies honest-proposal consistency.

\begin{lemma}[Honest-proposal integrity]
\label{lem:smr:proposal-integrity}
Suppose slot $s$ satisfies honest-proposal consistency and its
\pcaccstr instance satisfies Availability.
If a publicly accepted output reaches the coordinate of an honest
proposer $p$, then it contains $\hash(P_{p,s})$ at that coordinate.
\end{lemma}
\begin{proof}
Proposer $p$ inserts its own signed proposal object $P_{p,s}$ before
broadcasting it.
Honest-proposal consistency therefore makes every honest input equal to
$\hash(P_{p,s})$ at $p$'s coordinate.
Availability makes the output value at that coordinate equal to a value
in an honest input, hence to $\hash(P_{p,s})$.
Collision resistance and canonical encoding identify the corresponding
proposal object as $P_{p,s}$.
\end{proof}

\begin{lemma}[Ranking agreement]\label{lem:smr:ranking}
For every slot $s$, all honest parties use the same
$\rank^{\sf MC}_s$ and the same instance identifier $\id_s$.
\end{lemma}
\begin{proof}
The initial ranking is public.
Assume inductively that the claim holds for slot $s$.
ASPC Agreement gives all honest parties the same high value and
Availability gives them the same canonical proposal-object preimages.
They therefore extract the same valid previous-slot low set, compute the
same blame results, and apply the deterministic two-step ranking update.
Thus they obtain the same $\rank^{\sf MC}_{s+1}$ and
$\id_{s+1}=\hash(s+1,1,\rank^{\sf MC}_{s+1})$.
\end{proof}

\begin{theorem}\label{thm:smr:1}
The \smr Protocol (\Cref{alg:smr}) solves \smr
(\Cref{def:smr}) under partial synchrony: Agreement always holds, and
Termination holds after GST.
\end{theorem}
\begin{proof}
Fix a slot $s$.
By \Cref{lem:smr:ranking}, all honest parties invoke the same ASPC
instance with the same coordinate-to-proposer mapping.
Its high outputs agree, so
${\sf Commit}(s,\rank^{\sf MC}_s,v_s^{\sf high})$ verifies and commits the
same ordered proposal objects at every honest party.
Every early low is a prefix of the agreed high by VPC Proof Soundness, so
early commits cannot conflict with finalization.
Hence the finalized vectors $\mathbf b^{\sf out}_{i,s}$ agree.

After GST, the proposal deadline causes every correct party to invoke the
slot ASPC instance.
ASPC Termination gives every correct party its high output, after which it
finalizes slot $s$, computes the common next ranking, and starts slot
$s+1$.
Induction over slots proves Termination.
\end{proof}

\begin{theorem}\label{thm:smr:2}
The \smr Protocol (\Cref{alg:smr}) achieves $f$-Censorship Resistance
(\Cref{def:censor}).
\end{theorem}
\begin{proof}
Consider a censored post-GST slot $s$.
Some honest proposal $P_{p,s}$ is absent from the agreed high.
By \Cref{lem:smr:proposal-integrity}, the high cannot reach $p$'s
coordinate with another value; it must truncate earlier.
ASPC Accountability therefore returns
\[
b_s:={\sf Blame}^{\sf SPC}_{\id_s}
      (v_s^{\sf high},\pi_s^{\sf high})
\]
for a Byzantine proposer $b_s$.

The second ranking-update step moves $b_s$ to the absolute end, after all
honest proposers.
Subsequent valid low evidence blames only Byzantine proposers and cannot
move an honest proposer behind $b_s$.
Once $b_s$ is behind all honest proposers, a truncation at its coordinate
cannot omit an honest proposal, so it cannot be charged for another
censored slot.
Every censored slot is therefore charged to a distinct Byzantine
proposer.
There are at most $f$ such proposers, proving the claim.
\end{proof}

\begin{theorem}\label{thm:smr:leaderless}
The \smr Protocol (\Cref{alg:smr}) achieves Leaderless Termination
(\Cref{def:leaderless}) against one suspended party per round and up to
$f-1$ Byzantine parties.
\end{theorem}
\begin{proof}
Fix a slot $s$.
Its proposal collection is a synchronized logical round.
Every non-suspended honest proposer broadcasts to all correct parties,
while reliable messages missed by a suspended correct party remain pending
until it resumes.
Consequently every correct party eventually freezes its buffer and invokes
the same slot ASPC instance with the common ranking.

By \Cref{thm:pcstr:leaderless}, the ASPC instance terminates despite one
mobile suspension per round and produces an agreed high.
The dispersal layer supplies every certified proposal-object preimage, so
commit and evidence extraction terminate.
By \Cref{lem:smr:ranking}, all correct parties compute the same next
ranking.
These are local finite steps, after which each correct party invokes
${\sf NewSlot}(s+1)$.
Induction over slots proves Leaderless Termination.
\end{proof}

\begin{theorem}\label{thm:smr:complexity}
The \smr Protocol has amortized commit latency
\[
\tau_{\sf prop}+3\delta.
\]
Thus it has amortized latency $2\Delta+3\delta$ after GST,
$\Delta+3\delta$ under synchronized slot starts, and $4\delta$ in a
failure-free execution with synchronized starts.

In the good case---under synchrony and with an honest first proposer in
ASPC's view-2 ranking---it has $O(n^2)$ messages and
\[
O((\hashsize+\sigsize)n^3+cn^2)
\]
communication per slot.
In the worst case it has $O(n^3)$ messages and
\[
O((\hashsize+\sigsize)n^4+cn^2)
\]
communication per slot, where $c$ is the proposal-payload size.
\end{theorem}
\begin{proof}
\emph{Amortized latency.}
Call a non-censored slot \emph{slow} if some honest proposal is absent
from an honest party's low and is therefore committed only with the high.
By \Cref{lem:smr:proposal-integrity}, that low is truncated.
Its APC proof blames a Byzantine proposer, and the honest party carries
the evidence in its next proposal.

If the next slot is non-censored, its agreed high includes every honest
proposal object and therefore commits this evidence.
The first ranking-update step demotes the blamed proposer.
The source slot and its carrier can both be slow before this update takes
effect, so each successful Byzantine blame accounts for at most two slow
non-censored slots.
If the carrier slot is censored, charge the lost evidence to that censored
slot instead.
By \Cref{thm:smr:2}, there are at most $f$ censored slots, and at most
$f$ first effective demotions of Byzantine proposers.
Hence only $O(f)$ non-censored slots are slow.

Every remaining non-censored slot commits all honest proposals when ASPC
outputs its view-1 low, after proposal collection plus three VPC rounds.
Finite exceptional slots vanish in the amortized limit, giving
$\tau_{\sf prop}+3\delta$.
Substituting the three proposal-collection cases yields the stated bounds.

\emph{Complexity.}
Proposal dissemination costs $O(n^2)$ messages and $O(cn^2)$ payload
communication.
Each proposal carries at most one optimized low proof of size
$O((\hashsize+\sigsize)n)$, contributing
$O((\hashsize+\sigsize)n^3)$ communication.
The ASPC instance costs $O(n^2)$ messages and
$O((\hashsize+\sigsize)n^3)$ communication in the good case, and
$O(n^3)$ messages and
$O((\hashsize+\sigsize)n^4)$ communication in the worst case
(\Cref{thm:pcstr:complexity}).
Adding these terms proves the claim.
\end{proof}

\subsection{Improving Slot Latency}
\label{sec:smr:slot-latency}

Sequential slots advance only after ASPC outputs the agreed high.
Combining the proposal-collection window with
\Cref{thm:pcstr:complexity} gives the conservative slot-latency bound
\[
\tau_{\sf prop}
+2(f+1)\Delta+(5f+7)\delta.
\]
In a failure-free synchronized execution, proposal collection takes
$\delta$ and ASPC takes $7\delta$, giving slot latency $8\delta$.

\paragraph{Demoting parties that delay internal views.}
The $f$-dependent term can also be reduced on an amortized basis by
extending the outer accountability mechanism to later ASPC views.
Suppose a stabilized view $w>1$ satisfies honest-proposal consistency and
its VPC instance uses the concrete accountable construction.
A publicly verifiable truncated low $(v,\pi)$ then blames
\[
\rank^{\sf SPC}_{w}[|v|+1],
\]
which is Byzantine by APC Accountability; in particular, a zero-length
low blames the first party in that view's cyclic ranking.
Honest parties could carry such internal-view evidence into subsequent
slots and apply its demotion only after an agreed high includes it, using
the same canonical ordering and deduplication as for view-1 low evidence.
Because every internal ranking is a deterministic shift of the common
outer ranking, all honest parties would apply the same update.

With this extension, each Byzantine party can occupy a latency-critical
position only until its first committed blame.
Once the relevant Byzantine parties have been moved behind the honest
parties, view~2 starts with honest parties and the corresponding future
slots avoid the $f$ extra views.
This optimization does not improve the worst-case latency of the slots
that first expose those parties, and it is not part of
Algorithm~\ref{alg:smr} or the bounds proved above.

As in prior pipelined BFT protocols
\cite{arun2024shoal++,tonkikh2025raptr}, independent SMR instances can be
staggered to hide this sequential high-output delay; this optimization is
orthogonal to Algorithm~\ref{alg:smr}.

  \section{Communication-Optimized \pc Protocol}
\label{sec:pc:optimized}

This section gives a compact implementation of the threshold-based
\pc protocol from \Cref{alg:pc}.
It preserves the three-round structure and also provides the
verifiability and accountability properties of
\Cref{sec:pcver,sec:accountability}.

\subsection{Cryptographic Tools and Notation}
\label{sec:pc:optimized:notations}

Let
\[
q:=n-f,
\qquad
t:=\left\lfloor q/2\right\rfloor+1.
\]
We instantiate a batchable $t$-out-of-$n$ threshold signature with
threshold BLS~\cite{boldyreva2003threshold}, established
by a fixed DKG for the committee.
Ordinary vote signatures use registered BLS keys with proof of possession,
so same-message multisignatures are rogue-key safe.
Party $i$ can produce a publicly verifiable share
$\mathsf{ShareSign}_i(m)$; any $t$ valid shares on the same message combine
into a constant-size threshold signature $\mathsf{TSig}(m)$.
Threshold signatures on distinct messages can themselves be batched using
standard BLS aggregation~\cite{boneh2003aggregate} into one constant-size,
rogue-key-safe aggregate.
For repeated messages, signatures are first combined by message and the
resulting distinct-message signatures are batched.
The same indexed batching convention applies to verifiable threshold
shares in stopping proofs.
Verification binds every signature/share to a registered committee index
and its complete message.
All messages bind the protocol instance, phase, and coordinate.
For complexity, party identifiers and logical lengths are machine words;
their explicit $O(n)$ metadata is included below.

The threshold-signature assumption avoids carrying one signer bitmap per
coordinate.
Without it, the construction remains correct using ordinary aggregate
signatures, but a QC may require $O(nL)$ signer metadata.

\subsection{Protocol Description}
\label{sec:pc:optimized:protocol}

\paragraph{First-round vote.}
Party $i$ signs every coordinate of its input:
\begin{equation}\label{eq:pc:vote1}
\voteone_i=
\left(
v_i^{\sf in},
\{\mathsf{ShareSign}_i(\id,1,v_i^{\sf in}[1]),\ldots,
  \mathsf{ShareSign}_i(\id,L,v_i^{\sf in}[L])\}
\right).
\end{equation}
Its size is $O(\const L+\sigsize L)$.

\paragraph{First-round certificate.}
After receiving \voteone messages from a set $S$ of $q$ distinct parties,
a party computes $\x={\sf QC1Certify}(S)$ as in \Cref{alg:pc}.
For every $k\le|\x|$, it combines $t$ shares on
$(\id,k,\x[k])$ into a threshold signature $\tau_k$.
The signatures $\tau_1,\ldots,\tau_{|\x|}$ are batched into
$\mathcal{T}_{\x}$.

If $|\x|<L$, let $k:=|\x|+1$ and define
\[
\begin{aligned}
V_k &:= \{(i,v_i^{\sf in}[k]):i\in S\},\\
\Sigma_k^{\sf stop}
    &:= \mathsf{Aggregate}
        \{\mathsf{ShareSign}_i(\id,k,v_i^{\sf in}[k]):i\in S\},\\
N_k &:= (V_k,\Sigma_k^{\sf stop}).
\end{aligned}
\]
The aggregate in $N_k$ batch-verifies all listed shares.
They prove that no value at the cut coordinate has $t$ supporters.
If $|\x|=L$, set $N_{|\x|+1}:=\bot$.
Thus
\begin{equation}\label{eq:pc:qc1}
\qcone=(\x,\mathcal{T}_{\x},N_{|\x|+1}).
\end{equation}

To verify \qcone, a party checks the batched threshold signatures for all
coordinates of $\x$.
If $\x$ is short, it also checks that $N_{|\x|+1}$ contains valid shares
from $q$ distinct parties and that no value occurs $t$ times.
For public verification, support signatures for different coordinates
need not use the same $q$-party set as the stopping proof; the analysis
below relies only on threshold support for included coordinates and a
valid $q$-party no-support proof at the first excluded coordinate.
The vector $\x$ costs $O(\const L)$, its batched threshold signature costs
$O(\sigsize)$, and the optional stopping proof lists $q=O(n)$ indexed
values plus one aggregate signature, costing
$O(\const n+n+\sigsize)$.
Hence
\[
|\qcone|=O(\const(L+n)+n+\sigsize).
\]

\paragraph{Second-round vote and certificate.}
As in the basic protocol, each party signs every prefix of its certified
$\x$ and broadcasts
\begin{equation}\label{eq:pc:vote2}
\votetwo_i=
(\x,\sig_i(\prefix(\id,\x)),\qcone).
\end{equation}
Here
\[
\sig_i(\prefix(\id,\x))
:=
\{\sig_i(\id,k,\x[1:k]):0\le k\le|\x|\}
\]
is an ordered bundle of independently verifiable, domain-separated prefix
signatures, including the empty prefix for $k=0$.
From $q$ valid \votetwo messages, let
\[
\xmcp:=\mcp{\{\x\in\qctwo\}}.
\]
Let $S_2$ be the $q$ distinct \votetwo senders.
The compact certificate contains their aggregate prefix signature
$\Sigma_2$ on $\xmcp$ and a boundary proof $B_2$ showing that the common
prefix cannot extend one more coordinate.
Every vote carried in $B_2$ must be a valid authenticated \votetwo from a
sender in $S_2$.
If some $\x$ equals $\xmcp$, $B_2$ contains that vote and its valid
stopping \qcone.
Otherwise, $B_2$ contains two such votes whose certified $\x$ values both
extend $\xmcp$ and differ at coordinate $|\xmcp|+1$.
For a full-length $\xmcp$, set $B_2:=\bot$.
Hence
\begin{equation}\label{eq:pc:qc2}
\qctwo=(\xmcp,S_2,\Sigma_2,B_2).
\end{equation}
A verifier checks that $|S_2|=q$, that $\Sigma_2$ contains a valid prefix
signature by every member of $S_2$, and that the sender-bound proof $B_2$
proves maximality of $\xmcp$.
Each second-round vote contains one length-$L$ vector, an $O(\sigsize L)$
prefix-signature bundle, and one \qcone, for total size
$O(\const(L+n)+n+\sigsize L)$.
The remaining QC2 fields---$\xmcp$, $S_2$, and $\Sigma_2$---cost
$O(\const L+n+\sigsize)$, while $B_2$ contains at most two complete
second-round votes.
Therefore
\[
|\qctwo|=O(\const(L+n)+n+\sigsize L).
\]

\paragraph{Third-round vote and output certificate.}
Parties broadcast
\begin{equation}\label{eq:pc:vote3}
\votethree_i=(\xmcp,\sig_i(\id,\xmcp),\qctwo).
\end{equation}
The $\xmcp$ values in any valid third-round quorum are mutually
consistent.
Let $\xmcpmcp$ and $\xmcpmce$ be the shortest and longest such values.
Let $S_3$ be the $q$ distinct \votethree senders.
The output proof retains the complete authenticated votes
$W_{\min}$ and $W_{\max}$ from members of $S_3$ that certify the shortest
and longest values, together with a batched aggregate of all third-round
signatures:
\begin{equation}\label{eq:pc:qc3}
\qcthree=
(\xmcpmcp,W_{\min},
 \xmcpmce,W_{\max},S_3,\{\ell_i\}_{i\in S_3},\Sigma_3).
\end{equation}
Here $\ell_i$ is the logical length of party $i$'s signed $\xmcp$.
The verifier checks that $|S_3|=q$, reconstructs every signed value as the
length-$\ell_i$ prefix of $\xmcpmce$, verifies $\Sigma_3$, checks that
$\xmcpmcp$ and $\xmcpmce$ are the minimum- and maximum-length values, and
verifies that $W_{\min}$ and $W_{\max}$ are valid votes from the
corresponding members of $S_3$ with valid boundary \qctwo certificates.
Each third-round vote contains $\xmcp$, one signature, and one \qctwo, so
its size is $O(\const(L+n)+n+\sigsize L)$.
The two endpoint vectors, signer set, logical lengths, and aggregate
signature outside the boundary votes cost only
$O(\const L+n+\sigsize)$.
Since \qcthree retains only the two complete boundary votes
$W_{\min},W_{\max}$,
\[
|\qcthree|=O(\const(L+n)+n+\sigsize L).
\]
The protocol outputs
\[
v^{\sf low}:=\xmcpmcp,
\qquad
v^{\sf high}:=\xmcpmce,
\qquad
\pi^{\sf low}=\pi^{\sf high}:=\qcthree.
\]

\subsection{Analysis}
\label{sec:pc:optimized:analysis}

\begin{lemma}\label{lem:pc:optimized:qc1}
Every valid optimized \qcone certifies an $\x$ such that:
\begin{enumerate}
    \item every coordinate $k\le|\x|$ has support from at least $t$
    distinct \voteone senders; and
    \item if $|\x|<L$, no value at coordinate $|\x|+1$ has $t$
    supporters in the certified stopping quorum.
\end{enumerate}
\end{lemma}
\begin{proof}
The first property follows from unforgeability of the threshold signatures
in $\mathcal{T}_{\x}$.
The second follows by verifying all $q$ distinct shares in
$N_{|\x|+1}$ and counting their signed values.
\end{proof}

\begin{lemma}\label{lem:pc:optimized:validity}
Let $v_{\cH}=\mcp{\{v_h^{\sf in}\}_{h\in\cH}}$.
Then $v_{\cH}\preceq\x$ for every valid optimized \qcone.
\end{lemma}
\begin{proof}
Proceed by induction on $k\le|v_{\cH}|$.
All honest parties sign the same value $v_{\cH}[k]$ at coordinate $k$.
Any stopping quorum of size $q$ contains at least
$n-2f\ge t$ honest parties, so it cannot certify no support at $k$.
Moreover, no different value can obtain a threshold signature because
$t>f$.
Thus, after appending $v_{\cH}[1:k-1]$, every valid $\x$ also appends
$v_{\cH}[k]$.
\end{proof}

\begin{lemma}[Optimized availability]
\label{lem:pc:optimized:availability}
Every coordinate of a publicly verifiable optimized low or high output
appears at the same coordinate in an honest party's input.
\end{lemma}
\begin{proof}
Every coordinate of a valid $\x$ has a threshold signature with
$t>f$ shares, at least one of which was produced by an honest party with
that input value.
A valid \qctwo has prefix signatures from $q$ parties, including an honest
party that signs only prefixes of an $\x$ obtained from a valid \qcone.
Thus every coordinate of its $\xmcp$ traces to a threshold-supported
coordinate of an honest party's certified $\x$.
Finally, $W_{\min}$ and $W_{\max}$ in a valid \qcthree contain valid
\qctwo certificates for the low and high values, respectively.
\end{proof}

\begin{lemma}[Optimized accountable truncation]
\label{lem:pc:optimized:accountable}
In an execution satisfying honest-proposal consistency, if a publicly
verifiable optimized low or high output $v$ has $|v|<L$, then
$\rank_{\id}[|v|+1]$ is Byzantine.
\end{lemma}
\begin{proof}
Let $k=|v|+1$ and suppose $p_k=\rank_{\id}[k]$ is honest.
All honest parties sign the same value $z$ at coordinate $k$.
Hence no valid stopping proof can show that a $q$-party quorum has no
threshold-supported value at $k$.
Moreover, any threshold-supported value at $k$ has an honest supporter
because $t>f$, so two distinct threshold-supported values would contradict
honest-proposal consistency.

The authenticated boundary vote $W_{\min}$ or $W_{\max}$ for $v$ contains
a valid \qctwo with a boundary proof at $k$.
That proof must be either a stopping proof or two conflicting
threshold-certified values.
Both cases are impossible by the preceding argument, a contradiction.
\end{proof}

\paragraph{Public verification.}
Let ${\sf ValidQC3}^{\sf opt}(\id,\qcthree)$ check all optimized QC1/QC2
conditions above, the $q$ distinct third-round signers and their recorded
lengths, the aggregate $\Sigma_3$, and the authenticated boundary votes
$W_{\min},W_{\max}$.
Define
\[
\begin{aligned}
\mathcal F_{\id}^{\sf low}(v,\qcthree)
&\Longleftrightarrow
{\sf ValidQC3}^{\sf opt}(\id,\qcthree)
\wedge v=\qcthree.\xmcpmcp,\\
\mathcal F_{\id}^{\sf high}(v,\qcthree)
&\Longleftrightarrow
{\sf ValidQC3}^{\sf opt}(\id,\qcthree)
\wedge v=\qcthree.\xmcpmce.
\end{aligned}
\]

\begin{theorem}\label{thm:pc:optimized-correct}
The optimized protocol solves \pc and \pcver.
For executions satisfying honest-proposal consistency, it also solves
\pcacc.
\end{theorem}
\begin{proof}
Termination follows because every honest party eventually receives $q$
votes per round and all certification procedures are finite.
Validity follows from \Cref{lem:pc:optimized:validity}.

The QC2 consistency and QC3 Prefix-property arguments are identical to
\Cref{lem:pc:1,thm:pc:1}: two quorums intersect in an honest signer that
sends only one vote per phase.
In a valid \qctwo, $\Sigma_2$ proves that $q$ parties signed the claimed
common prefix, while the sender-bound proof $B_2$ proves that it cannot be
extended.
The corresponding $S_3$, length, and boundary-vote checks establish the
shortest and longest values in a valid \qcthree.
This proves \pc.

Every honest output certificate passes ${\sf ValidQC3}^{\sf opt}$, giving
VPC Proof Completeness.
For Proof Soundness, consider any accepted low certificate \qcthree and
high certificate $\qcthree'$ for the same instance.
Their signer sets intersect in an honest party's unique \votethree value
$\xmcp$.
The verified minimum and maximum lengths imply
\[
\qcthree.\xmcpmcp
\preceq\xmcp
\preceq\qcthree'.\xmcpmce.
\]
Validity of every accepted low follows from
\Cref{lem:pc:optimized:validity}.
Thus the predicates implement \pcver.

Availability follows from
\Cref{lem:pc:optimized:availability}, and accountability follows from
\Cref{lem:pc:optimized:accountable}.
\end{proof}

\begin{theorem}\label{thm:pc:optimized:complexity}
The optimized protocol has round complexity $3$, message complexity
$O(n^2)$, and communication complexity
\[
O(\const n^2L+\const n^3+\sigsize n^2L).
\]
\end{theorem}
\begin{proof}
The protocol has three all-to-all voting rounds.
Each party sends one message to every party per round, giving
$O(n^2)$ messages in total.

In round one,
\[
|\voteone|
=O(\const L+\sigsize L),
\]
because the vote contains $L$ values and $L$ threshold-signature shares.
Thus round-one communication is
\[
O(\const n^2L+\sigsize n^2L).
\]

By the size calculation above,
\[
|\qcone|=O(\const(L+n)+n+\sigsize).
\]
A \votetwo additionally contains the vector $\x$ and its
$O(\sigsize L)$ prefix-signature bundle, so
\[
|\votetwo|
=O(\const(L+n)+n+\sigsize L).
\]
Hence round two costs
\[
O(\const n^2L+\const n^3+n^3+\sigsize n^2L).
\]

Similarly, a \votethree contains $\xmcp$, one signature, and a \qctwo.
Since
\[
|\qctwo|=O(\const(L+n)+n+\sigsize L),
\]
round three has the same asymptotic cost as round two.
The two boundary votes in \qctwo and \qcthree contribute only constant
factors, and batching keeps aggregate signatures constant-size.

Under the machine-word convention, the $O(n^3)$ identifier and length
metadata is absorbed by $O(\const n^3)$.
Summing the three rounds gives
\[
O(\const n^2L+\const n^3+\sigsize n^2L).
\]
\end{proof}

  \section{Discussion}
\label{sec:discussion}

In this section, we discuss several related topics and extensions suggested by our results.

\subsection{2-Round \pc under $n\geq 5f+1$}
\label{sec:discussion:2-round}

In~\Cref{sec:pc:5f+1} we present a 2-round \pc protocol under the stronger
resilience assumption $n\ge 5f+1$.
From each first-round quorum it derives the coordinate-wise prefix
supported by at least $n-2f$ votes.
Two conflicting certified coordinates would have support sets
intersecting in at least $n-4f\ge f+1$ parties and hence in an honest
party, which is impossible.
The resulting certified prefixes are therefore mutually consistent, so a
single exchange of them suffices to derive both outputs.
The same threshold rule also makes the construction accountable
(\Cref{thm:pc:5f+1:accountable}).

Note that there remains a gap between our upper and lower bounds for 2-round \pc. On the one hand, our lower bound (\Cref{thm:lower-bound}) shows that 2-round \pc is impossible when $n\leq 4f$. On the other hand, our 2-round construction requires the stronger condition $n\geq 5f+1$. 
Closing this gap is an interesting open problem.

\subsection{Relation to Graded Consensus}
\label{sec:discussion:graded}

To show the relation between \pc and graded consensus~\cite{connected-cons-2024}, we consider \pc protocols that additionally satisfy Consistency and Availability (\Cref{def:pc:availability}).
Our \pc construction satisfies both properties.

\begin{definition}[\pccon]\label{def:pccon}
    \pccon satisfies all conditions of \pc (\Cref{def:pc}), and in addition ensures: 
    \begin{itemize}
        \item Consistency: $v_{i}^{\sf high}\sim v_{j}^{\sf high}$ for any $i,j\in \cH$.
    \end{itemize}
\end{definition}

\begin{theorem}\label{thm:pccon:1}
    \pc Protocol (\Cref{alg:pc}) solves \pccon. 
\end{theorem}

\begin{proof}
    By~\Cref{lem:pc:1}, $\qctwo.\xmcp\sim\qctwo'.\xmcp$ for any $\qctwo, \qctwo'$. 
    Then, by definition, $\qcthree.\xmcpmce=\mce{\{\xmcp: (\xmcp, *, *) \in \qcthree\}}=\qctwo.\xmcp$ for some \qctwo, and $\qcthree'.\xmcpmce=\qctwo'.\xmcp$ for some $\qctwo'$.
    Therefore, $\qcthree.\xmcpmce\sim \qcthree'.\xmcpmce$ for any $\qcthree,\qcthree'$. 

    According to the protocol, for any $i,j\in \cH$, $v_{i}^{\sf high}=\qcthree.\xmcpmce$ and $v_{j}^{\sf high}=\qcthree'.\xmcpmce$ for some $\qcthree,\qcthree'$. 
    Therefore $v_{i}^{\sf high}\sim v_{j}^{\sf high}$ for any $i,j\in \cH$.
\end{proof}

The communication-optimized implementation of
\Cref{sec:pc:optimized} preserves the same certified $\xmcp$ values and
selects the same shortest and longest endpoints.
Hence the preceding argument also establishes Consistency for that
implementation; it satisfies Availability by
\Cref{lem:pc:optimized:availability}.

Graded consensus is a well-studied primitive, defined for the multi-valued case in~\cite{connected-cons-2024} as follows:

\begin{definition}[Graded Consensus]
    In a graded consensus protocol, each party has an input $v$, and parties output pairs $(x,g)$ where $g\in\{0,1,2\}$ and $x=\bot$ if and only if $g=0$. A graded consensus protocol has the following properties:
    \begin{itemize}
        \item Agreement. If two honest parties output pairs $(x,g),(x',g')$, then $|g-g'|\leq 1$ and if $x\neq x'$, either $x=\bot$ or $x'=\bot$. 
        \item Validity. The output $x$ must be the input of at least one honest party or $\bot$. Moreover, if all honest parties have the same input $v$, then all honest parties output $(v,2)$.
        \item Termination. All honest parties eventually output.
    \end{itemize}
\end{definition}

\subsubsection{Reduction from \pc with Consistency and Availability to Graded Consensus.}
\label{subsec:gc:pc-to-gc}

Given a protocol $\Pi_{\sf PC}$ that solves \pccon and satisfies
Availability (\Cref{def:pc:availability}), we can construct a protocol
$\Pi_{\sf GC}$ that solves graded consensus as follows.
To execute a graded consensus instance with input $v_{i}^{\sf GC}$, 
each party $i$ runs $\Pi_{\sf PC}$ on a length-1 vector input $[v_{i}^{\sf GC}]$. 
When $\Pi_{\sf PC}$ outputs $(v_{i}^{\sf low}, v_{i}^{\sf high})$, 
party $i$ determines its output for $\Pi_{\sf GC}$ as:
\begin{itemize}
    \item Output $(\bot, 0)$ if $v_{i}^{\sf low}= v_{i}^{\sf high} = []$ (empty vector);
    \item Output $(x, 1)$ if $v_{i}^{\sf low}= []$ and $v_{i}^{\sf high} = [x]$;
    \item Output $(x, 2)$ if $v_{i}^{\sf low}= v_{i}^{\sf high} = [x]$.
\end{itemize}

Agreement and Termination follow directly from \pccon.
Availability ensures that every non-$\bot$ output $x$ is the input of an
honest party, while \pc Validity ensures grade~2 when all honest inputs
are equal.
Thus the construction satisfies Graded Consensus Validity.
Our basic and communication-optimized \pc constructions satisfy
Availability
(\Cref{lem:availability,lem:pc:optimized:availability}).
The reduction does not introduce any additional latency or communication
overhead.
Therefore, our \pc protocol, which satisfies Consistency and Availability,
yields a graded consensus protocol with
$O(n^2)$ message complexity, 
$O(\sigsize n^2 + n^3)$ communication complexity (for constant size input), 
and a latency of three rounds (\Cref{thm:pc:optimized:complexity}).
This improves the round count relative to the seven-round optimally
resilient unauthenticated construction of~\cite{connected-cons-2024},
although that construction additionally satisfies Binding.


\subsubsection{Reduction from Graded Consensus to \pc with Consistency and Availability.}
\label{subsec:gc:gc-to-pc}
To execute a \pccon instance with a length-$L$ input vector $v_{i}^{\sf in}$, 
each party $i$ runs $L$ independent instances of \emph{graded consensus} in parallel. 
In the $k$-th instance, party $i$ inputs the element $v_{i}^{\sf in}[k]$. 
Once all graded consensus instances terminate, let 
\[
g_i = (g_i[1], \ldots, g_i[L])
\quad \text{and} \quad 
v_{i}^{\sf GC} = (v_{i}^{\sf GC}[1], \ldots, v_{i}^{\sf GC}[L])
\]
denote, respectively, the vectors of grades and decided values produced by these instances. 
Set $g_i[L+1]:=0$.
Party $i$ then produces its $\Pi_{\sf PC}$ outputs as follows:
\begin{itemize}
    \item 
    $v_{i}^{\sf low}$ is the \emph{longest prefix} of $v_{i}^{\sf GC}$ for which all corresponding grades are $2$; that is,
    \[
    v_{i}^{\sf low} = [v_{i}^{\sf GC}[1], \ldots, v_{i}^{\sf GC}[k]]
    \quad \text{where } g_i[j] = 2, \forall j\leq k, \; g_i[k+1] < 2.
    \]
    \item 
    $v_{i}^{\sf high}$ is the \emph{longest prefix} of $v_{i}^{\sf GC}$ for which all corresponding grades are at least $1$; that is,
    \[
    v_{i}^{\sf high} = [v_{i}^{\sf GC}[1], \ldots, v_{i}^{\sf GC}[k']]
    \quad \text{where } g_i[j] \geq 1, \forall j\leq k', \; g_i[k'+1] < 1.
    \]
\end{itemize}

Termination follows from termination of the $L$ parallel instances.
For every coordinate in the common prefix of the honest input vectors, all
honest parties input the same value; Graded Consensus Validity makes every
honest party output that value with grade~2, proving \pc Validity.
If an honest party's low contains coordinate $j$, its grade at every
coordinate through $j$ is~2.
Graded Consensus Agreement implies that every other honest party has the
same values with grades at least~1, so its high contains that low.
This proves the Prefix property.
Finally, whenever two honest highs both contain a coordinate, their
non-$\bot$ values agree; hence the high vectors are prefix-consistent.
Graded Consensus Validity also ensures that each non-$\bot$ output
coordinate is the input of an honest party in the corresponding instance,
which proves Availability.
Thus the construction solves \pc with Consistency and Availability.

This reduction introduces no additional latency.
Hence, the lower bound from \Cref{sec:lower-bound} also applies to Graded Consensus and proves that the Graded Consensus algorithm from \Cref{subsec:gc:pc-to-gc} is latency-optimal for $n \le 4f$.

\subsection{Relation to Order Fairness}
\label{sec:discussion:order-fairness}

Our notion of censorship resistance provides bounded inclusion rather than
formal order fairness~\cite{kelkar2020order}.
Transactions assigned to distinct non-censored slots are included in slot
order, but our definition does not constrain transaction-to-slot
assignment or ordering within a slot.
With an explicit admission rule, short slot times can bound practical
inclusion delays and ordering inversions, complementing dedicated fairness
mechanisms~\cite{kelkar2020order,kelkar2023themis,cachin2022quick}.
Recent BFT SMR systems have pushed slot times toward a single message
delay~\cite{doidge2024moonshot}, and production blockchains deploy similar
techniques to drive slot times below 50ms~\cite{aptos-50-ms}.

  \section{2-Round \pc under $n\geq 5f+1$}\label{sec:pc:5f+1}

The protocol is presented in \Cref{alg:pc:5f+1} and closely mirrors the
structure of the 3-round \pc protocol (\Cref{alg:pc}).
From a first-round quorum, it certifies the coordinate-wise threshold
prefix supported by at least $n-2f$ votes.
When $n\ge5f+1$, two conflicting threshold values would require support
sets whose intersection contains an honest party, contradicting that an
honest party sends only one first-round vote.
Thus all certified prefixes are mutually consistent, and parties can
derive both outputs directly from a single second-round quorum.

The validation and instance-binding conventions of \Cref{sec:pc:simple}
apply: certificates contain messages from distinct senders, nested
certificates and claimed values are recursively checked, and signatures
bind the phase and instance identifier.
The instance identifier also commits to the proposer ranking when the
protocol is used as \pcacc.
As shown below, the existing \qctwo is a proof for both outputs, and the
coordinate-wise threshold rule makes the protocol accountable without
additional rounds or messages.

\begin{algorithm}[t!]
\caption{2-Round \pc Protocol under $n\geq 5f+1$}
\label{alg:pc:5f+1}
\begin{algorithmic}[1]
\Function{\sf QC1Certify}{\qcone}
    \State $\x := \tpref{n-2f}{\{v_i:v_i\in\qcone\}}$
        \label{line:5f+1:x}
    \Statex\Comment{Append a coordinate exactly when one value has support from at least $n-2f$ votes.}
    \State \Return \x
\EndFunction

\Statex 

\Function{\sf QC2Certify}{\qctwo}
    \State $\xmcp:=\mcp{\{\x \in \qctwo\}}$
    \Statex\Comment{\xmcp is the maximum common prefix of all \x votes in \qctwo.}
    \State $\xmce:=\mce{\{\x \in \qctwo\}}$
    \Statex\Comment{\xmce is the minimum common extension of all \x votes in \qctwo.}
    \State \Return $(\xmcp, \xmce)$
\EndFunction

\Statex 

\Upon{input $v_{i}^{\sf in}$}
    \State broadcast $\voteone:=(v_{i}^{\sf in}, \sig_i(v_{i}^{\sf in}))$
\EndUpon

\Statex 

\Upon{first obtaining $n-f$ valid \voteone from distinct parties}
    \State let \qcone be the set of $n-f$ \voteone messages
    \State $\x:={\sf QC1Certify}(\qcone)$
    \State broadcast $\votetwo:=(\x, \sig_i(\x), \qcone)$
\EndUpon

\Statex 

\Upon{first obtaining $n-f$ valid \votetwo from distinct parties}
    \State let \qctwo be the set of $n-f$ \votetwo messages
    \State $(\xmcp, \xmce):={\sf QC2Certify}(\qctwo)$
    \State {\bf output} $v_{i}^{\sf low}:=\xmcp$ and
        $v_{i}^{\sf high}:=\xmce$, retaining \qctwo as the proof of both
\EndUpon

\end{algorithmic}
\end{algorithm}

\subsection{Analysis}\label{sec:pc:5f+1:analysis}

For brevity, we use the notation $\qcone.\x:={\sf QC1Certify}(\qcone)$ and $(\qctwo.\xmcp, \qctwo.\xmce):={\sf QC2Certify}(\qctwo)$.

\begin{lemma}[Certified-prefix consistency]\label{lem:pc:5f+1:1}
For any valid first-round certificates $\qcone,\qcone'$,
\[
\qcone.\x\sim\qcone'.\x.
\]
\end{lemma}
\begin{proof}
Suppose toward contradiction that the two threshold prefixes first
conflict at coordinate $k$, with values $z\ne z'$.
Let $S$ and $S'$ be the senders supporting $z$ and $z'$ at coordinate
$k$ in $\qcone$ and $\qcone'$, respectively.
Then
\[
|S\cap S'|
\ge 2(n-2f)-n
=n-4f
\ge f+1.
\]
The intersection contains an honest party, which would have signed two
different values at the same coordinate in the same instance.
This contradicts the one-vote-per-phase rule.
\end{proof}

\begin{lemma}[Certified validity]\label{lem:pc:5f+1:2}
For every valid \qcone,
\[
\mcp{\{v_h^{\sf in}\}_{h\in\cH}}\preceq\qcone.\x.
\]
\end{lemma}
\begin{proof}
Let $v_{\cH}:=\mcp{\{v_h^{\sf in}\}_{h\in\cH}}$.
A \qcone contains $n-f$ distinct votes, at least $n-2f$ of which are
honest.
At every coordinate of $v_{\cH}$, these honest votes contain the same
value.
That value therefore has threshold support $n-2f$, so
${\sf QC1Certify}$ appends every coordinate of $v_{\cH}$.
\end{proof}

\begin{theorem}\label{thm:pc:5f+1:1}
Algorithm~\ref{alg:pc:5f+1} solves \pc (\Cref{def:pc}) under $n\ge 5f+1$.
\end{theorem}
\begin{proof}
\emph{Prefix.}
Any two second-round certificates $\qctwo,\qctwo'$ contain $n-f$
distinct senders each, so their sender sets intersect in at least $n-2f$
parties and hence in an honest party.
Let $\x$ be that party's unique \votetwo value.
Then
\[
\qctwo.\xmcp\preceq\x\preceq\qctwo'.\xmce.
\]
The right-hand MCE exists because
\Cref{lem:pc:5f+1:1} makes all valid \votetwo values pairwise consistent.
This proves the Prefix property for every pair of honest outputs.

\emph{Termination.}
The $n-f$ honest parties suffice to form both certificates.
Reliable channels eventually deliver their votes, and both certification
functions terminate on finite inputs.
Thus every honest party outputs.

\emph{Validity.}
By \Cref{lem:pc:5f+1:2}, every valid \votetwo value extends
$v_{\cH}:=\mcp{\{v_h^{\sf in}\}_{h\in\cH}}$.
Their maximum common prefix therefore also extends $v_{\cH}$, proving
Validity.
\end{proof}

\paragraph{Public verification.}
Let ${\sf ValidQC2}(\id,\qctwo)$ recursively check that \qctwo contains
$n-f$ valid \votetwo messages from distinct senders for instance $\id$,
that every nested \qcone contains $n-f$ valid \voteone messages from
distinct senders, and that each claimed $\x$ equals
${\sf QC1Certify}(\qcone)$.
Define
\[
\begin{aligned}
\mathcal F_{\id}^{\sf low}(v,\qctwo)
&\Longleftrightarrow
{\sf ValidQC2}(\id,\qctwo)
\wedge (v,*)={\sf QC2Certify}(\qctwo),\\
\mathcal F_{\id}^{\sf high}(v,\qctwo)
&\Longleftrightarrow
{\sf ValidQC2}(\id,\qctwo)
\wedge (*,v)={\sf QC2Certify}(\qctwo).
\end{aligned}
\]

\begin{theorem}\label{thm:pc:5f+1:accountable}
Under $n\ge5f+1$, \Cref{alg:pc:5f+1}, augmented with its existing
\qctwo output proof, solves \pcacc.
It adds no rounds, messages, or communication.
\end{theorem}
\begin{proof}
The protocol solves \pc by \Cref{thm:pc:5f+1:1}.
Every honest output is computed from a valid \qctwo, proving VPC Proof
Completeness.
Recursive validation lets the arguments for Prefix and Validity above
apply to any two publicly accepted certificates, proving Proof Soundness.
Thus the construction solves \pcver.

It remains to prove Accountability.
Consider an execution satisfying honest-proposal consistency and a
publicly accepted truncated low or high output $v$.
Let $k=|v|+1$ and suppose for contradiction that
$p_k=\rank_{\id}[k]$ is honest.
Every honest first-round vote then contains the same value $z$ at
coordinate $k$.

By \Cref{lem:pc:5f+1:1}, the certified $\x$ values in a valid \qctwo are
totally ordered by the prefix relation.
Consequently, its low and high outputs are respectively the shortest and
longest certified values, so either output is an actual certified $\x$.
Let \qcone be the nested certificate for the certified value equal to
$v$.
It contains at least $n-2f$ honest votes, all carrying $z$ at coordinate
$k$.
Since $v$ reaches coordinate $k-1$, the threshold-prefix rule must append
$z$ at coordinate $k$, contradicting $|v|=k-1$.
Therefore the first excluded proposer is Byzantine, exactly as required
by \Cref{def:pcacc}.

The proof is the already assembled \qctwo, so no protocol communication
is added.
\end{proof}

\begin{theorem}\label{thm:pc:5f+1:complexity}
\Cref{alg:pc:5f+1} has round complexity $2$, message complexity
$O(n^2)$, and communication complexity
$O((\const L+\sigsize)n^3)$.
\end{theorem}
\begin{proof}
There are two all-to-all voting rounds, hence $O(n^2)$ messages.
A \voteone has size $O(\const L+\sigsize)$.
A \qcone, and therefore a \votetwo carrying one, has size
$O((\const L+\sigsize)n)$.
Multiplying the second-round message size by its $O(n^2)$ deliveries
gives $O((\const L+\sigsize)n^3)$ total communication; the first round is
smaller.
The locally assembled \qctwo proof has size
$O((\const L+\sigsize)n^2)$ and requires no additional transmission.
\end{proof}

\fi

\end{document}